\documentclass[11pt, a4paper]{article}
\usepackage{jheppub}
\usepackage{amsmath, amssymb, amsfonts, amsthm, bbm}
\usepackage{enumitem}
\usepackage{pgf,tikz, graphicx}

\definecolor{jred}{rgb}{0.8,0,0}
\definecolor{jgreen}{rgb}{0,0.7,0}
\definecolor{jblue}{rgb}{0,0,0.8}

\usetikzlibrary{arrows,patterns,shapes,positioning,fit}
\tikzstyle{c} =	[coordinate]
\tikzstyle{vc} = [circle, draw=black, line width=1pt, inner sep=0pt, minimum size=2mm]
\tikzstyle{v} =  [circle, draw=black, line width=.2pt, fill=black, inner sep=0pt, minimum size=1.5mm]
\tikzstyle{vb} =[circle, draw=black, line width=.1pt, fill=jred, inner sep=0pt, minimum size=1mm]
\tikzstyle{vh} =[circle, draw=black, line width=.1pt, fill=jblue, inner sep=0pt, minimum size=2mm]
\tikzstyle{vs} =[coordinate]
\tikzstyle{e} =	[draw=jred,line width=1pt]
\tikzstyle{eb}= [draw=jgreen,line width=1pt]
\tikzstyle{es}= [dashed, line width=1pt]
\tikzstyle{ev}= [draw=jgreen,line width=1pt]
\tikzstyle{f} = [line width=0.01pt,dotted,fill=blue, fill opacity=.1]
\tikzstyle{cs}= [ellipse, draw=black, line width=.1pt, fill=jred, inner sep=1pt, minimum size=2mm]

\newcommand{\edgevg}{
\begin{tikzpicture}[x=2ex,y=2ex,baseline={([yshift=-.59ex]current bounding box.center)}]
\node [vb]	(1)	at (1,0)	{};
\node [vb]	(2)	at (-1,0)	{};
\path
(1)	edge [ev,bend left=30]	(2)
(1)	edge [ev,bend right=30]	(2);
\end{tikzpicture}
}

\newcommand{\edgevgr}{
\begin{tikzpicture}[x=2ex,y=2ex,baseline={([yshift=-.59ex]current bounding box.center)}]
\node [vb]	(1)	at (1,0)	{};
\node [vb]	(2)	at (-1,0)	{};
\node at (0,0) {...};
\path
(1)	edge [ev,bend left=50] 	(2)
(1)	edge [ev,bend left=30] 	(2)
(1)	edge [ev,bend right=30] (2)
(1)	edge [ev,bend right=50]	(2);
\end{tikzpicture}
}

\newcommand{\mapd}{
\begin{tikzpicture}[x=2ex,y=2ex,baseline={([yshift=-.59ex]current bounding box.center)}]
\node [v]	(v1)at (0,0) 	{};
\node [v]	(v2)at (2.4,0) 	{};
\node [vb]	(1)	at (-1,0)	{};
\node [vb]	(2)	at (.6,-1)	{};
\node [vb]	(3)	at (.6,1)	{};
\node [vb]	(4)	at (1.4,0)	{};
\node [vb]	(5)	at (3,-1)	{};
\node [vb]	(6)	at (3,1)	{};
\path
\foreach \i in {1,2,3}{
    (\i) edge [e] (v1)
    }
\foreach \i in {4,5,6}{
    (\i) edge [e] (v2)
    }
\foreach \i/\j in {1/2,2/3,3/1,4/5,5/6,6/4}{
    (\i) edge [ev] (\j)
    }
\foreach \i/\j in {2/5,3/6}{
    (\i) edge [es] (\j)
    };
\end{tikzpicture}
}

\newcommand{\maptorus}{
\begin{tikzpicture}[x=2ex,y=2ex,baseline={([yshift=-.59ex]current bounding box.center)}]
\node [v]	(v3)	at (0,0) 	{};
\node [v]	(v4)	at (2,0) 	{};
\node [vb]	(11)	at (-1,0)	{};
\node [vb]	(12)	at (.6,-1)	{};
\node [vb]	(13)	at (.6,1)	{};
\node [vb]	(9)	at (3,0)	{};
\node [vb]	(8)	at (1.4,1)	{};
\node [vb]	(7)	at (1.4,-1)	{};
\begin{scope}[yshift=5ex]
\node [v]	(v1)	at (0,0) 	{};
\node [v]	(v2)	at (2.4,0) 	{};
\node [vb]	(1)	at (-1,0)	{};
\node [vb]	(2)	at (.6,-1)	{};
\node [vb]	(3)	at (.6,1)	{};
\node [vb]	(4)	at (1.4,0)	{};
\node [vb]	(5)	at (3,-1)	{};
\node [vb]	(6)	at (3,1)	{};
\end{scope}
\path
\foreach \i/\j in {1/v1,2/v1,3/v1,4/v2,5/v2,6/v2,11/v3,12/v3,13/v3,7/v4,8/v4,9/v4}{
    (\i) edge [e] (\j)
    }
\foreach \i/\j in {1/2,2/3,3/1,4/5,5/6,6/4,7/8,8/9,9/7,11/12,12/13,13/11}{
    (\i) edge [ev] (\j)
    }
\foreach \i/\j in {2/5,3/6,1/11,12/7,13/8}{
    (\i) edge [es] (\j)
    }
(4) edge [es,sloped] (9);
\end{tikzpicture}
}

\newcommand{\mapcon}{
\begin{tikzpicture}[x=2ex,y=2ex,baseline={([yshift=-.59ex]current bounding box.center)}]
\draw [ev] (0,1.8) circle (.5ex); 
\draw [ev] (2,1.8) circle (.5ex); 
\node [v]	(v3)at (0,0) 	{};
\node [v]	(v4)at (2,0) 	{};
\node [vb]	(1)	at (-1,0)	{};
\node [vb]	(2)	at (.6,-1)	{};
\node [vb]	(3)	at (.6,1)	{};
\node [vb]	(9)	at (3,0)	{};
\node [vb]	(8)	at (1.4,1)	{};
\node [vb]	(7)	at (1.4,-1)	{};
\node [v]	(v2)at (1,1.5) 	{};
\node [vb]	(4)	at (0,1.5)	{};
\node [vb]	(5)	at (2,1.5)	{};
\path
\foreach \i/\j in {4/v2,5/v2,1/v3,2/v3,3/v3,7/v4,8/v4,9/v4}{
    (\i) edge [e] (\j)
    }
\foreach \i/\j in {1/2,2/3,3/1,7/8,8/9,9/7}{
    (\i) edge [ev] (\j)
    }
\foreach \i/\j in {2/7,3/8,1/4,9/5}{
    (\i) edge [es] (\j)
    };
\end{tikzpicture}
}

\newcommand{\cvtvg}{
\begin{tikzpicture}[x=2ex,y=2ex,baseline={([yshift=-.59ex]current bounding box.center)}] 
\node [vb]	(1)	at (0,1)	{};
\node [vb]	(2)	at (0,-1)	{};
\path
\foreach \i/\j in {1/2,2/1}{
   (\i)	edge [ev,bend left=15]	(\j)
   (\i)	edge [ev,bend left=40]	(\j)
   };
\end{tikzpicture}
}

\newcommand{\cvt}{
\begin{tikzpicture}[x=2ex,y=2ex,baseline={([yshift=-.59ex]current bounding box.center)}] 
\node [v]	(v)	at (0,0)	{};
\node [vb]	(1)	at (1,0)	{};
\node [vb]	(2)	at (-1,0)	{};
\path
\foreach \i in {1,2}{
    (\i) edge [e] (v)
    }
\foreach \i/\j in {1/2,2/1}{
   (\i)	edge [ev,bend left=30]	(\j)
   (\i)	edge [ev,bend left=50]	(\j)
   };
\end{tikzpicture}
}

\newcommand{\cvfl}{
\begin{tikzpicture}[x=2ex,y=2ex,baseline={([yshift=-.59ex]current bounding box.center)}] 
\scriptsize{
\node [v]	(1)	at (0,0)	{};
\node [vb, label=above:1]	(b1)	at (-1,1)	{};
\node [vb, label=below:2]	(w1)	at (-1,-1)	{};
\node [vb, label=below:7]	(b2)	at (1,-1)	{};
\node [vb, label=above:8]	(w2)	at (1,1)	{};
\path
\foreach \i/\j in {1/2}{
   	(w\i) edge [ev]   (b\j)
	(b\i) edge [ev]  node [c, label=above:$c$] {} (w\j)
	}
\foreach \i in {1,2}{
	(w\i)   edge 	[ev]				(b\i)
	(w\i)	edge [ev,bend left=30]	(b\i)
	(w\i)	edge [ev,bend right=30]	(b\i)
	}
\foreach \i in {w1,b1,w2,b2}{
   	(\i) edge [e] (1)
	};
}
\end{tikzpicture}
}

\newcommand{\cvft}{
\begin{tikzpicture}[x=2ex,y=2ex,baseline={([yshift=-.59ex]current bounding box.center)}] 
\node [v]	(v)	at (0,0) 	{};
\node [vb]	(1)	at (-1,1)	{};
\node [vb]	(2)	at (-1,-1)	{};
\node [vb]	(3)	at (1,-1)	{};
\node [vb]	(4)	at (1,1)	{};
\path
\foreach \i in {1,2,3,4}{
    (\i) edge [e] (v)
    }
\foreach \i/\j in {1/2,2/1,3/4,4/3}{
   (\i)	edge [ev,bend left=15]	(\j)
   (\i)	edge [ev,bend left=40]	(\j)
   };
\end{tikzpicture}
}

\newcommand{\cvftl}{
\begin{tikzpicture}[x=2ex,y=2ex,baseline={([yshift=-.59ex]current bounding box.center)}] 
\scriptsize{
\node [v]	(v)	at (0,0) 	{};
\node [vb, label=above:1]	(1)	at (-1,1)	{};
\node [vb, label=below:2]	(2)	at (-1,-1)	{};
\node [vb, label=below:7]	(3)	at (1,-1)	{};
\node [vb, label=above:8]	(4)	at (1,1)	{};
\path
\foreach \i in {1,2,3,4}{
    (\i) edge [e] (v)
    }
\foreach \i/\j in {1/2,2/1,3/4,4/3}{
   (\i)	edge [ev,bend left=15]	(\j)
   (\i)	edge [ev,bend left=40]	(\j)
   };
}
\end{tikzpicture}
}

\newcommand{\cfishl}{
\begin{tikzpicture}[x=2ex,y=2ex,baseline={([yshift=-.59ex]current bounding box.center)}] 
\scriptsize{
\node [v]	(1)	at (0,0)	{};
\node [v]	(2)	at (4,0)	{};
\node [vb, label=above:1]	(b1)	at (-1,1)	{};
\node [vb, label=below:2]	(w1)	at (-1,-1)	{};
\node [vb, label=below:3]	(b2)	at (1,-1)	{};
\node [vb, label=above:4]	(w2)	at (1,1)	{};
\node [vb, label=above:5]	(b3)	at (3,1)	{};
\node [vb, label=below:6]	(w3)	at (3,-1)	{};
\node [vb, label=below:7]	(b4)	at (5,-1)	{};
\node [vb, label=above:8]	(w4)	at (5,1)	{};
\path
\foreach \i/\j in {1/2}{
   	(w\i) edge [ev]   (b\j)
	(b\i) edge [ev]  node [c, label=above:$c_1$] {} (w\j)
	}
\foreach \i/\j in {3/4}{
   	(w\i) edge [ev]   (b\j)
	(b\i) edge [ev]  node [c, label=above:$c_2$] {} (w\j)
	}
\foreach \i in {1,2,3,4}{
	(w\i) edge 	[ev]				(b\i)
	(w\i)	edge [ev,bend left=30]	(b\i)
	(w\i)	edge [ev,bend right=30]	(b\i)
	}
\foreach \i in {w1,b1,w2,b2}{
   	(\i) edge [e] (1)
	}
\foreach \i in {w3,b3,w4,b4}{
   	(\i) edge [e] (2)
	}	
\foreach \i/\j in {w2/b3,w3/b2}{
   	(\i) edge [es] (\j)
	};
}
\end{tikzpicture}
}

\newcommand{\cfishlc}{
\begin{tikzpicture}[x=2ex,y=2ex,baseline={([yshift=-.59ex]current bounding box.center)}] 
\scriptsize{
\node [v]	(1)	at (0,0)	{};
\node [v]	(2)	at (4,0)	{};
\node [vb, label=above:1]	(b1)	at (-1,1)	{};
\node [vb, label=below:2]	(w1)	at (-1,-1)	{};
\node [vb, label=below:3]	(b2)	at (1,-1)	{};
\node [vb, label=above:4]	(w2)	at (1,1)	{};
\node [vb, label=above:5]	(b3)	at (3,1)	{};
\node [vb, label=below:6]	(w3)	at (3,-1)	{};
\node [vb, label=below:7]	(b4)	at (5,-1)	{};
\node [vb, label=above:8]	(w4)	at (5,1)	{};
\path
\foreach \i/\j in {1/2,3/4}{
   	(w\i) edge [ev]   (b\j)
	(b\i) edge [ev]  node [c, label=above:$c$] {} (w\j)
    }
\foreach \i in {1,2,3,4}{
	(w\i) edge 	[ev]				(b\i)
	(w\i)	edge [ev,bend left=30]	(b\i)
	(w\i)	edge [ev,bend right=30]	(b\i)
	}
\foreach \i in {w1,b1,w2,b2}{
   	(\i) edge [e] (1)
	}
\foreach \i in {w3,b3,w4,b4}{
   	(\i) edge [e] (2)
	}	
\foreach \i/\j in {w2/b3,w3/b2}{
   	(\i) edge [es] (\j)
	};
}
\end{tikzpicture}
}

\newcommand{\cfish}{
\begin{tikzpicture}[x=2ex,y=2ex,baseline={([yshift=-.59ex]current bounding box.center)}] 
\scriptsize{
\node [v]	(1)	at (0,0)	{};
\node [v]	(2)	at (4,0)	{};
\node [vb]	(w1)	at (-1,-1)	{};
\node [vb]	(b1)	at (-1,1)	{};
\node [vb]	(w2)	at (1,1)	{};
\node [vb]	(b2)	at (1,-1)	{};
\node [vb]	(w3)	at (3,-1)	{};
\node [vb]	(b3)	at (3,1)	{};
\node [vb]	(w4)	at (5,1)	{};
\node [vb]	(b4)	at (5,-1)	{};
\path
\foreach \i/\j in {1/2,3/4}{
   	(w\i) edge [ev]   (b\j)
	(b\i) edge [ev]  node [c] {} (w\j)
	}
\foreach \i in {1,2,3,4}{
	(w\i) edge 	[ev]				(b\i)
	(w\i)	edge [ev,bend left=30]	(b\i)
	(w\i)	edge [ev,bend right=30]	(b\i)
	}
\foreach \i in {w1,b1,w2,b2}{
   	(\i) edge [e] (1)
	}
\foreach \i in {w3,b3,w4,b4}{
   	(\i) edge [e] (2)
	}	
\foreach \i/\j in {w2/b3,w3/b2}{
   	(\i) edge [es] (\j)
	};
}
\end{tikzpicture}
}

\newcommand{\cfisha}{
\begin{tikzpicture}[x=2ex,y=2ex,baseline={([yshift=-.59ex]current bounding box.center)}] 
\scriptsize{
\node [v]	(1)	at (0,0)	{};
\node [v]	(2)	at (4,0)	{};
\node [vb]	(w1)	at (-1,-1)	{};
\node [vb]	(b1)	at (-1,1)	{};
\node [vb]	(w2)	at (1,1)	{};
\node [vb]	(b2)	at (1,-1)	{};
\node [vb]	(w3)	at (3,-1)	{};
\node [vb]	(b3)	at (3,1)	{};
\node [vb]	(w4)	at (5,1)	{};
\node [vb]	(b4)	at (5,-1)	{};
\path
\foreach \i/\j in {1/2,3/4}{
   	(w\i) edge [ev]   (b\j)
	(b\i) edge [ev]  node [c] {} (w\j)
	}
\foreach \i in {1,2,3,4}{
	(w\i) edge 	[ev]				(b\i)
	(w\i)	edge [ev,bend left=30]	(b\i)
	(w\i)	edge [ev,bend right=30]	(b\i)
	}
\foreach \i in {w1,b1,w2,b2}{
   	(\i) edge [e] (1)
	}
\foreach \i in {w3,b3,w4,b4}{
   	(\i) edge [e] (2)
	}	
\foreach \i/\j in {w3/b2}{
   	(\i) edge [es] (\j)
	};
}
\end{tikzpicture}
}

\newcommand{\cfishb}{
\begin{tikzpicture}[x=2ex,y=2ex,baseline={([yshift=-.59ex]current bounding box.center)}] 
\scriptsize{
\node [v]	(1)	at (0,0)	{};
\node [v]	(2)	at (4,0)	{};
\node [vb]	(w1)	at (-1,-1)	{};
\node [vb]	(b1)	at (-1,1)	{};
\node [vb]	(w2)	at (1,1)	{};
\node [vb]	(b2)	at (1,-1)	{};
\node [vb]	(w3)	at (3,-1)	{};
\node [vb]	(b3)	at (3,1)	{};
\node [vb]	(w4)	at (5,1)	{};
\node [vb]	(b4)	at (5,-1)	{};
\path
\foreach \i/\j in {1/2,3/4}{
   	(w\i) edge [ev]   (b\j)
	(b\i) edge [ev]  node [c] {} (w\j)
	}
\foreach \i in {1,2,3,4}{
	(w\i) edge 	[ev]				(b\i)
	(w\i)	edge [ev,bend left=30]	(b\i)
	(w\i)	edge [ev,bend right=30]	(b\i)
	}
\foreach \i in {w1,b1,w2,b2}{
   	(\i) edge [e] (1)
	}
\foreach \i in {w3,b3,w4,b4}{
   	(\i) edge [e] (2)
	}	
\foreach \i/\j in {w2/b3}{
   	(\i) edge [es] (\j)
	};
}
\end{tikzpicture}
}

\newcommand{\cfishc}{
\begin{tikzpicture}[x=2ex,y=2ex,baseline={([yshift=-.59ex]current bounding box.center)}] 
\scriptsize{
\node [v]	(1)	at (0,0)	{};
\node [v]	(2)	at (4,0)	{};
\node [vb]	(w1)	at (-1,-1)	{};
\node [vb]	(b1)	at (-1,1)	{};
\node [vb]	(w2)	at (1,1)	{};
\node [vb]	(b2)	at (1,-1)	{};
\node [vb]	(w3)	at (3,-1)	{};
\node [vb]	(b3)	at (3,1)	{};
\node [vb]	(w4)	at (5,1)	{};
\node [vb]	(b4)	at (5,-1)	{};
\path
\foreach \i/\j in {1/2,3/4}{
   	(w\i) edge [ev]   (b\j)
	(b\i) edge [ev]  node [c] {} (w\j)
	}
\foreach \i in {1,2,3,4}{
	(w\i) edge 	[ev]				(b\i)
	(w\i)	edge [ev,bend left=30]	(b\i)
	(w\i)	edge [ev,bend right=30]	(b\i)
	}
\foreach \i in {w1,b1,w2,b2}{
   	(\i) edge [e] (1)
	}
\foreach \i in {w3,b3,w4,b4}{
   	(\i) edge [e] (2)
	};
}
\end{tikzpicture}
}

\newcommand{\cvsla}{
\begin{tikzpicture}[x=2ex,y=2ex,baseline={([yshift=0ex]current bounding box.center)}]
\scriptsize{
\node [v]	(1)	at (0,0)	{};
\begin{scope}[rotate=120]
\node [vb, label=left:1]		(b1)	at (.7,1.2)	{};
\node [vb, label=below:2]		(w1)	at (-.7,1.2)	{};
\end{scope}
\begin{scope}[rotate=240]
\node [vb, label=below:7]		(b2)	at (.7,1.2)	{};
\node [vb, label=right:8]		(w2)	at (-.7,1.2)	{};
\end{scope}
\begin{scope}[rotate=0]
\node [vb, label=right:5]		(b3)	at (.7,1.2)	{};
\node [vb, label=left:4]		(w3)	at (-.7,1.2)	{};
\end{scope}	
\path
(w1) edge [ev]  node [c, label=below:$c$] {} (b2)
\foreach \i/\j in {2/3,3/1}{
   	(w\i) edge [ev]   (b\j)
	}
\foreach \i in {1,2,3}{
	(w\i) edge 	[ev]				(b\i)
	(w\i)	edge [ev,bend left=30]	(b\i)
	(w\i)	edge [ev,bend right=30]	(b\i)
	}
\foreach \i in {w1,b1,w2,b2,w3,b3}{
   	(\i) edge [e] (1)
	}
(w3) edge [es, bend left=90] (b3);
}
\end{tikzpicture}
}

\newcommand{\cvslb}{
\begin{tikzpicture}[x=2ex,y=2ex,baseline={([yshift=-1.2ex]current bounding box.center)}]
\scriptsize{
\node [v]	(1)	at (0,0)	{};
\begin{scope}[rotate=60]
\node [vb, label=above:1]		(b1)	at (.7,1.2)	{};
\node [vb, label=left:2]		(w1)	at (-.7,1.2)	{};
\end{scope}
\begin{scope}[rotate=180]
\node [vb, label=left:3]		(b2)	at (.7,1.2)	{};
\node [vb, label=right:6]		(w2)	at (-.7,1.2)	{};
\end{scope}
\begin{scope}[rotate=300]
\node [vb, label=right:7]		(b3)	at (.7,1.2)	{};
\node [vb, label=above:8]		(w3)	at (-.7,1.2)	{};
\end{scope}	
\path
(b1) edge [ev]  node [c, label=above:$c$] {} (w3)
\foreach \i/\j in {1/2,2/3}{
   	(w\i) edge [ev]   (b\j)
	}
\foreach \i in {1,2,3}{
	(w\i) edge 	[ev]				node 	{}	(b\i)
	(w\i)	edge [ev,bend left=30]	node 	{}	(b\i)
	(w\i)	edge [ev,bend right=30]	node 	{}	(b\i)
	}
\foreach \i in {w1,b1,w2,b2,w3,b3}{
   	(\i) edge [e] (1)
	}
(w2) edge [es, bend left=90] (b2);
}
\end{tikzpicture}
}

\newcommand{\cvslc}{
\begin{tikzpicture}[x=2ex,y=2ex,baseline={([yshift=-.59ex]current bounding box.center)}]
\scriptsize{
\node [v]	(1)	at (0,0)	{};
\begin{scope}[rotate=90]
\node [vb, label=left:1]		(b1)	at (.7,1.2)	{};
\node [vb, label=left:2]		(w1)	at (-.7,1.2)	{};
\end{scope}
\begin{scope}[rotate=210]
\node [vb, label=below:5]		(b2)	at (.7,1.2)	{};
\node [vb, label=right:8]		(w3)	at (-.7,1.2)	{};
\end{scope}
\begin{scope}[rotate=330]
\node [vb, label=right:7]		(b3)	at (.7,1.2)	{};
\node [vb, label=above:4]		(w2)	at (-.7,1.2)	{};
\end{scope}	
\path
(b1) edge [ev]  node [c, label=150:$c_1$] {} (w2)
(w2) edge [ev]  node [c, label=30:$c_2$] {} (b3)
\foreach \i/\j in {1/1,1/2,3/2,3/3}{
   	(w\i) edge [ev]   (b\j)
	}
\foreach \i in {1,2,3}{
	(w\i)	edge [ev,bend left=30]	node 	{}	(b\i)
	(w\i)	edge [ev,bend right=30]	node 	{}	(b\i)
	}
\foreach \i in {w1,b1,w2,b2,w3,b3}{
   	(\i) edge [e] (1)
	}
(w2) edge [es, bend left=140] (b2);
}
\end{tikzpicture}
}

\newcommand{\cvsld}{
\begin{tikzpicture}[x=2ex,y=2ex,baseline={([yshift=-.59ex]current bounding box.center)}]
\scriptsize{
\node [v]	(1)	at (0,0)	{};
\begin{scope}[rotate=90]
\node [vb, label=left:1]		(b1)	at (.7,1.2)	{};
\node [vb, label=left:2]		(w1)	at (-.7,1.2)	{};
\end{scope}
\begin{scope}[rotate=210]
\node [vb, label=below:6]		(b2)	at (.7,1.2)	{};
\node [vb, label=right:8]		(w3)	at (-.7,1.2)	{};
\end{scope}
\begin{scope}[rotate=330]
\node [vb, label=right:7]		(b3)	at (.7,1.2)	{};
\node [vb, label=above:3]		(w2)	at (-.7,1.2)	{};
\end{scope}	
\path
(b1) edge [ev]  node [c, label=150:$c_1$] {} (w2)
(w2) edge [ev]  node [c, label=30:$c_2$] {} (b3)
\foreach \i/\j in {1/1,1/2,3/2,3/3}{
   	(w\i) edge [ev]   (b\j)
	}
\foreach \i in {1,2,3}{
	(w\i)	edge [ev,bend left=30]	node 	{}	(b\i)
	(w\i)	edge [ev,bend right=30]	node 	{}	(b\i)
	}
\foreach \i in {w1,b1,w2,b2,w3,b3}{
   	(\i) edge [e] (1)
	}
(w2) edge [es, bend left=140] (b2);
}
\end{tikzpicture}
}

\newcommand{\cvsd}{
\begin{tikzpicture}[x=2ex,y=2ex,baseline={([yshift=-.59ex]current bounding box.center)}]
\scriptsize{
\node [v]	(1)	at (0,0)	{};
\begin{scope}[rotate=90]
\node [vb]		(b1)	at (.7,1.2)	{};
\node [vb]		(w1)	at (-.7,1.2)	{};
\end{scope}
\begin{scope}[rotate=210]
\node [vb]		(b2)	at (.7,1.2)	{};
\node [vb]		(w3)	at (-.7,1.2)	{};
\end{scope}
\begin{scope}[rotate=330]
\node [vb]		(b3)	at (.7,1.2)	{};
\node [vb]		(w2)	at (-.7,1.2)	{};
\end{scope}	
\path
\foreach \i/\j in {1/1,1/2,2/1,2/3,3/2,3/3}{
   	(w\i) edge [ev]   (b\j)
	}
\foreach \i in {1,2,3}{
	(w\i)	edge [ev,bend left=30]	node 	{}	(b\i)
	(w\i)	edge [ev,bend right=30]	node 	{}	(b\i)
	}
\foreach \i in {w1,b1,w2,b2,w3,b3}{
   	(\i) edge [e] (1)
	}
(w2) edge [es, bend left=140] (b2);
}
\end{tikzpicture}
}

\renewcommand{\[}{\begin{equation}}
\renewcommand{\]}{\end{equation}}

\newcommand{\boxd}[1]{\boxed{\phantom{\Biggl(}#1\phantom{\Biggl)}}}
\renewcommand{\eqref}[1]{Eq.\,(\ref{#1})}
\newtheoremstyle{mydef}{}{}{}{}{\bfseries}{.}{ }{\thmname{#1}\thmnumber{ #2}\thmnote{ (#3)}}
\theoremstyle{mydef}
\newtheorem{defin}{Definition}[section]
\newtheorem{remark}[defin]{Remark}
\newtheorem{example}[defin]{Example}
\theoremstyle{plain}
\newtheorem{proposition}[defin]{Proposition}
\newtheorem{theorem}[defin]{Theorem}

\newcommand{\cdef}[1]{\emph{#1}}
\def\d{\mathrm{d}}
\def\e{\textrm e}
\def\i{\mathrm i}
\def\ie{{i.e.}~}
\def\eg{{e.g.}~}

\newcommand{\N}{\mathbb N}
\newcommand{\Z}{\mathbb Z}
\newcommand{\Q}{\mathbb Q}
\newcommand{\In}{\subset}

\newcommand{\cp}{\mathcal{C}}		
\newcommand{\dm}{dim}				
\newcommand{\og}{\gamma}
\newcommand{\tg}{\Gamma}
\newcommand{\ttg}{\tilde{\Gamma}}
\newcommand{\sg}{\Theta}
\newcommand{\ogl}{g}
\newcommand{\tgl}{G}
\newcommand{\tgla}{{G}}
\newcommand{\tglb}{{G'}}
\newcommand{\sgl}{{H}}
\newcommand{\V}{\mathcal{V}}
\renewcommand{\H}{\mathcal{H}}
\newcommand{\Hx}{\mathcal{H}^{\mathrm{ext}}}
\newcommand{\E}{\mathcal{E}}
\newcommand{\Ex}{\mathcal{E}^{\mathrm{ext}}}
\renewcommand{\S}{\mathcal{S}}
\newcommand{\Sx}{\mathcal{S}^{\mathrm{ext}}}
\newcommand{\F}{\mathcal{F}}
\newcommand{\Fx}{\mathcal{F}^{\mathrm{ext}}}
\newcommand{\Fi}{\mathcal{F}^{\mathrm{int}}}
\renewcommand{\v}{\nu}
\newcommand{\h}{\mu}
\newcommand{\vs}{{\sigma_1}}
\newcommand{\nv}{V}

\newcommand{\nei}{E}

\newcommand{\nf}{F}

\newcommand{\nc}{K}
\newcommand{\ei}{\iota}
\newcommand{\es}{{\sigma_2}}
\newcommand{\VS}{\mathcal{S}^v}
\newcommand{\ES}{\mathcal{S}^e}
\newcommand{\vd}{d^v}
\newcommand{\hd}{d^h}
\newcommand{\vg}{g_{v}}
\newcommand{\vgb}{\beta_\mathrm{vg}}
\newcommand{\vgp}{\pi_\mathrm{vg}}
\newcommand{\cc}{\mathcal{K}}

\newcommand{\pset}[1]{\mathbf{2}^{#1}}
\newcommand{\pmset}[1]{\mathcal{P}(#1)}
\newcommand{\jv}{j_\V}
\newcommand{\jh}{j_\H}
\newcommand{\js}{j_\S}
\newcommand{\tsim}{\cong}
\newcommand{\osim}{\cong}

\newcommand{\OG}{\mathbf{G}_{1}}
\newcommand{\Vtypeg}{\mathbf{V}}

\newcommand{\TG}{\mathbf{G}_{2}}
\newcommand{\TGog}{\mathbf{G}_{2\textrm{c}}^\og}
\newcommand{\KG}{\mathbf{K}}
\newcommand{\PG}{\mathbf{P}}

\newcommand{\psd}{\mathbf{P}^{\textrm{s.d.}}}
\newcommand{\sgr}{\In}
\newcommand{\Res}{\mathbf{R}}
\newcommand{\Vtyper}{\mathbf{S}^*}
\newcommand{\res}{\mathrm{res}}
\newcommand{\br}{\partial}			    
\newcommand{\brt}{\widetilde{\partial}}		
\newcommand{\skl}{\mathrm{skl}}
\newcommand{\vgs}{\varsigma}            
\newcommand{\vgt}{\widetilde{\varsigma}\,}            
\newcommand{\ins}[2]{{\mathcal{I}}(#1,#2)}
\newcommand{\aut}{\mathrm{Aut}\,}

\newcommand{\uca}{\mathcal{A}}
\newcommand{\btg}{\mathcal{G}}

\newcommand{\hfd}{\mathcal{H}^{\textrm{f2g}}}

\newcommand{\one}{\mathbbm{1}}
\newcommand{\id}{\mathrm{id}}
\newcommand{\cop}{\Delta}
\newcommand{\cou}{\epsilon}
\newcommand{\conp}{*}
\newcommand{\anti}{S}

\newcommand{\rota}{R}

\newcommand{\w}{\omega}
\newcommand{\sdd}{\omega^{\textrm{sd}}}
\newcommand{\uvd}{d_{\rk}}
\newcommand{\ks}{{2\zeta}}          
 
\newcommand{\gdeg}{\omega^\textsc{g}}

\newcommand{\std}{D}				
\newcommand{\rk}{r}             	
\newcommand{\gd}{d}                 
\newcommand{\Sia}{S_\textsc{ia}}   		
\newcommand{\vx}{\pmb{x}}
\newcommand{\vp}{\pmb{p}}
\newcommand{\cn}[1]{\lambda_{#1}}

\begin{document}

\begin{flushright}
MaPhy-AvH/2021-02
\end{flushright}

\title
{Renormalization in combinatorially non-local field theories: the Hopf algebra of 2-graphs}

\author{Johannes Th\"urigen}
\affiliation{Mathematisches Institut der Westf\"alischen Wilhelms-Universit\"at M\"unster\\ Einsteinstr. 62, 48149 M\"unster, Germany, EU \\
Institut f\"ur Physik/Institut f\"ur Mathematik der Humboldt-Universit\"at zu Berlin\\
Unter den Linden 6, 10099 Berlin, Germany, EU}
\emailAdd{johannes.thuerigen@uni-muenster.de}

\begin{abstract}
{It is well known that the mathematical structure underlying renormalization in perturbative quantum field theory is based on a Hopf algebra of Feynman diagrams. 
A precondition for this is locality of the field theory. Consequently, one might suspect that non-local field theories such as matrix or tensor field theories cannot benefit from a similar algebraic understanding. 
Here I show that, on the contrary, the renormalization and perturbative diagramatics of a broad class of such field theories is based in the same way on a Hopf algebra. 
These theories are characterized by interaction vertices with graphs as external structure leading to Feynman diagrams which can be summed up under the concept of ``2-graphs".
From the renormalization perspective, such graph-like interactions are as much local as point-like interactions.
They differ in combinatorial details as I exemplify with the central identity 
for the perturbative series of combinatorial correlation functions.
This sets the stage for a systematic study of perturbative renormalization as well as non-perturbative aspects, \eg Dyson-Schwinger equations, for a number of combinatorially non-local field theories with possible applications to quantum gravity, statistical models and more.
}
\end{abstract}

\setcounter{tocdepth}{2}

\maketitle


\section{Locality in combinatorial non-local field theories}

Locality in QFT allows to perform perturbative renormalization order by order.
Subtracting divergences is described mathematically by a Hopf algebra of Feynman diagrams~\cite{Kreimer:1998te,Connes:2000km,Connes:2001hk,Kreimer:2002br}.
For a given diagram, the coproduct separates divergent subdiagrams.
Then the renormalization operation subtracting the divergences of this subdiagram is based on the Hopf algebra's antipode.
Combinatorially, one needs a closed set of diagrams in which it is possible to ``separate" divergent diagrams and counter them by single vertices with the same external structure.
It is well known that this works in many cases for local, \ie point-like, interactions \cite{Kreimer:1998te,Connes:2000km,Connes:2001hk, Kreimer:2002br, Kreimer:2006gg, vanSuijlekom:2007jv, Kreimer:2008jm, Borinsky:2018,Yeats:2008,Yeats:2017}.
But there are also examples of renormalizable field theories with certain ``non-local'' interactions for which a Connes-Kreimer type Hopf algebra can be found~\cite{Tanasa:2007xa,Tanasa:2013kh,Raasakka:2013wa,Avohou:2015co}.
More specifically, these are combinatorially non-local field theories such as non-commutative quantum field theory and matrix field theory \cite{Kontsevich:1992en,Grosse:2004bm, Grosse:2005ec,Wulkenhaar:2019,Hock:2020tg}, and its generalization to field theories of higher rank tensor fields \cite{BenGeloun:2013fw, BenGeloun:2017un, BenGeloun:2014gp}, including group field theory~\cite{Carrozza:2013uq,DePietri:2000dt,Oriti:2003uw,Freidel:2005jy}.
It is thus a natural question whether this is just a coincidence in specific examples or a more general feature, and whether the (Hopf algebra) mechanisms are different or the same.

Standard quantum field theories on a $\std$-dimensional space are local in the sense that they have point-like interactions, for example  $\phi(\vx)^n$ for a scalar field $\phi$ on coordinates $x_1,...,x_\std$.
This relates to energy-momentum conservation at interaction vertices via Fourier transformation
\[\label{eq:localinteraction}
\Sia[\phi] = 
\int\d\vx \,\cn{n}\phi(\vx)^n =
\cn{n}\int\d\vx \prod_{i=1}^n \int \d\vp_i \,\tilde{\phi}(\vp_i) \e^{\i \vp_i\cdot\vx} =
\cn{n}\int\prod_{i=1}^n \d\vp_i\,  \delta\bigg(\sum_{i=1}^n \vp_i\bigg) \prod_{i=1}^n \tilde{\phi}(\vp_i)
\]
where $\delta$ is the 
Dirac distribution. 
It constrains interactions to have incoming and outgoing energy-momentum being equal.
In particular, correlation functions and effective vertices are characterized by a given external constraint structure, as well as each Feynman diagram in their perturbative expansion.
For perturbatively renormalizable quantum field theories, this external structure of interaction vertices allows to subtract the divergent part of an amplitude identifying it in subgraphs.
Remarkably, the structure of this renormalization procedure, independent of a specific renormalization scheme, is captured by the Hopf algebra of Feynman diagrams~\cite{Kreimer:1998te, Connes:2000km,Connes:2001hk,Kreimer:2002br}. 
Thereby, the external structure of vertices induced by locality of the interactions plays a crucial role.

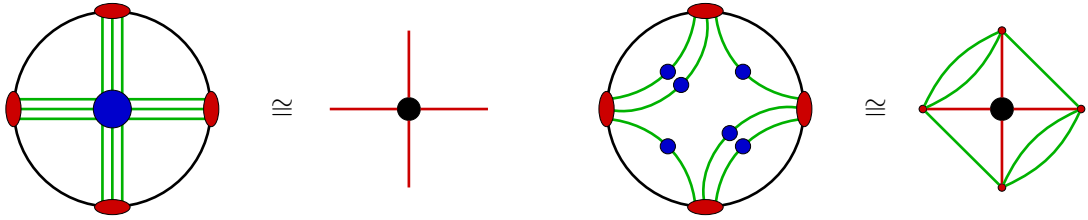
\begin{figure}
    \centering
    \begin{tikzpicture}[scale=1.3]
    \node [vs]	(11)	at (1,.1)	{};
    \node [vs]	(12)	at (1,0)	{};
    \node [vs]	(13)	at (1,-.1)	{};    
    \node [vs]	(21)	at (.1,1)	{};
    \node [vs]	(22)	at (0,1)	{};
    \node [vs]	(23)	at (-.1,1)	{};    
    \node [vs]	(31)	at (-1,.1)	{};
    \node [vs]	(32)	at (-1,0)	{};
    \node [vs]	(33)	at (-1,-.1)	{};    
    \node [vs]	(41)	at (.1,-1)	{};
    \node [vs]	(42)	at (0,-1)	{};
    \node [vs]	(43)	at (-.1,-1)	{};
    \node [vc, fit=(12) (32)] {};
    \path
    \foreach \i/\j in {11/31,12/32,13/33,21/41,22/42,23/43}{
    (\i) edge [ev] (\j)
    };
    \node [vh, minimum size=5mm]    {};
    \node [cs,fit=(11)(13)] {};
    \node [cs,fit=(21)(23)] {};
    \node [cs,fit=(31)(33)] {};
    \node [cs,fit=(41)(43)] {};
    
    \node [c, label=0:$\cong$] at (1.5,0) {};
    
    \begin{scope}[xshift=3cm, scale=.8]
    \node [vs]	(1)	at (1,0)	{};
    \node [vs]	(2)	at (0,1)	{};
    \node [vs]	(3)	at (-1,0)	{};
    \node [vs]	(4)	at (0,-1)	{};
    \node [v, fill=black, minimum size =3mm]   (0) at (0,0)    {};
    \path
    \foreach \i in {1,2,3,4}{
    (\i) edge [e] (0)
    };
    \end{scope}
    
    \begin{scope}[xshift=6cm]
    \node [vs]	(12)	at (1,.1)	{};
    \node [vs]	(13)	at (1,0)	{};
    \node [vs]	(14)	at (1,-.1)	{};    
    \node [vs]	(21)	at (.1,1)	{};
    \node [vs]	(24)	at (0,1)	{};
    \node [vs]	(23)	at (-.1,1)	{};    
    \node [vs]	(32)	at (-1,.1)	{};
    \node [vs]	(31)	at (-1,0)	{};
    \node [vs]	(34)	at (-1,-.1)	{};    
    \node [vs]	(41)	at (.1,-1)	{};
    \node [vs]	(42)	at (0,-1)	{};
    \node [vs]	(43)	at (-.1,-1)	{};
    \node [vc, fit=(13) (31)] {};
    \path
    \foreach \i/\j in {14/41,21/12,32/23,43/34}{
    (\i) edge [ev,bend right=40] node [vh] {} (\j)
    }
    \foreach \i/\j in {13/42, 31/24}{
    (\i) edge [ev,bend right=60] node [vh] {} (\j)
    };
    \node [cs,fit=(12) (14)] {};
    \node [cs,fit=(21) (23)] {};
    \node [cs,fit=(32) (34)] {};
    \node [cs,fit=(41) (43)] {};
    
    \node [c, label=0:$\cong$] at (1.5,0) {};
    
    \begin{scope}[xshift=3cm, scale=.8]
    \node [vb]	(1)	at (1,0)	{};
    \node [vb]	(2)	at (0,1)	{};
    \node [vb]	(3)	at (-1,0)	{};
    \node [vb]	(4)	at (0,-1)	{};
    \node [v, fill=black, minimum size =3mm]   (0) at (0,0)    {};
    \path
    (1) edge [ev] (2)
    (3) edge [ev] (4)
    \foreach \i in {1,2,3,4}{
    (\i) edge [e] (0)
    }
    \foreach \i/\j in {1/4,4/1,2/3,3/2}{
    (\i) edge [ev,bend right=20] node [c] {} (\j)
    };
    \end{scope}
    \end{scope}
    \end{tikzpicture}
    \caption{
    Comparison of the combinatorial structure of an order-$n$ interaction vertex in a combinatorially local theory (left) and non-local theory (right).
    In both cases, there are $n$ fields (red ellipses, here $n=4$) which depend on a number of arguments (green lines).
    In the local case, \eqref{eq:localinteraction}, due to Lorentz (or some other ``rotational'') symmetry these arguments form a vector which is constrained by a single delta distribution (blue circle)
    such that all of this structure is equivalently captured by a vertex (black) with $n$ half edges (red lines).
    On the other hand, in the non-local case, 
    \eqref{eq:nonlocalinteraction})
    the arguments 
    are convoluted pairwaise (blue dots);
    to capture this structure it is necessary to add a second set of edges (green lines) leading to the notion of a 2-graph.
    }
    \label{fig:nonlocalinteractions}
\end{figure}

The aim of this work is to show that also a broad class of ``non-local'' field theories is characterized by the same Hopf-algebraic structure of perturbative renormalization.
Thus, it is not locality in the sense of point-like interactions but an appropriate class of external structures which is crucial for perturbative renormalization.
The field theories under consideration are characterized by an external constraint structure which pairs single entries of ``momenta''.
For example, let $\phi$ again be a scalar field, a function of $\rk$ arguments $\vp=(p^1,...,p^\rk)$ where each argument $p^a$ is in a $\gd$-dimensional manifold.
Then, a generic interaction of order~$n$, with $n\cdot\rk$ even, has the form
\[\label{eq:nonlocalinteraction}
\Sia[\phi] = \cn{n }\int\prod_{i=1}^n \d\vp_i\,  \prod_{(ia,jb)}\delta(p_i^a-p_j^b) \prod_{i=1}^n \tilde{\phi}(\vp_i)
\]
where the product over pairs $(ia,jb)$ means that for each argument $p_i^a$ there is a convolution $\delta(p_i^a-p_j^b)$ with exactly one other argument $p_j^b$. 
As a consequence the diagrammatic representation of interactions is not just a vertex in a graph but has to capture this pairing of arguments (see Fig.~\ref{fig:nonlocalinteractions}). 
In effect, each interaction has the combinatorial structure of a graph itself, \ie the necessary class of external structures is graphs.

Consequently, the Feynman diagrams 
of a combinatorially non-local field theory are certain gluings of vertex graphs, or equivalently, they are graphs with the additional structure of ``strands'' at each edge.
Following a strand through a cycle (loop) of the graph gives a face.
Thus, such ``strand graphs'' are actually two-dimensional objects.
Indeed, they are two-dimensional combinatorial complexes in a specific sense (Prop. \ref{prop:complex}).
This is well known in the case of fields with $\rk=2$ arguments, that is matrices; these generate combinatorial maps, also called ribbon or fat graphs, which are dual to $n$-angulations of surfaces \cite{DiFrancesco:1995ih}.
With more arguments $\rk>2$, one can use additional structure to extend this duality to $n$-angulations of $\rk$-dimensional (pseudo) manifolds \cite{Gurau:2011dw,Gurau:2010iu,Gurau:2016}.
Here I want to consider such diagrams in full generality allowing vertices of arbitrary order $n$ with fields~$\phi_i$, $i=1,2,...,n$, with an arbitrary number of arguments $\rk_i$ (similar to \cite{Oriti:2015kv}). 
To emphasize that such Feynman diagrams are just a generalization of standard Feynman graphs adding a second layer I will call them \emph{2-graphs} and any field theory with such combinatorics a \emph{combinatorially non-local field theory} (cNLFT).

In this work I show that the Hopf-algebraic structure underlying renormalization in cNLFT is very general and independent of any specific theory,
along a similar logic as for local field theory \cite{Borinsky:2018}.
Renormalizablity of various cNLFTs is known as for example  Grosse/Wulkenhaar's non-commutative field theory \cite{Grosse:2004bm, Grosse:2005ec} related to Kontsevich's matrix model \cite{Kontsevich:1992en,Grosse:2018dc}, tensor-field models \cite{BenGeloun:2013fw, BenGeloun:2014gp} and group field theories \cite{Carrozza:2013uq}.
In any such case there is a set of superficially divergent 2-graphs and the Hopf algebra encodes the procedure of identifying (coproduct) and subtracting (antipode) these divergences in a given Feynman 2-graph via subgraph contraction. 
I will generalize the contraction operation from graphs to 2-graphs.
This gives rise to a general Hopf algebra of 2-graphs (Thm.\,\ref{th:hopfalgebra}).
One can then derive the renormalization Hopf algebra of any specific renormalizable theory 
(such as the known ones for the non-commutative field theory \cite{Tanasa:2007xa,Tanasa:2013kh} and some tensor-field models \cite{Raasakka:2013wa,Avohou:2015co})
as a subalgebra of the general Hopf algebra.

With the appropriate concept of 2-graphs as well as contraction and insertion operations at hand, the general algebraic structure of renormalization in cNLFT turns out to be exactly the same as for local field theory, but the algebras differ in combinatorial details. 
In particular, 2-graphs have more structure, and thus less symmetry, than usual graphs.
I show this for the central identity (Thm.\,\ref{th:centralidentity}), \ie the action of the coproduct on an infinite series over 2-graphs which captures the combinatorics of the perturbative expansion of Green's functions; this formula differs from the one of local field theory in that there occurs the symmetry factor of the boundary graph of each interaction vertex, not just the factorial~$n!$.
These concrete formula set the stage for explicit (BPHZ) renormalization \cite{Thurigen:2021wn} and the investigation of Dyson-Schwinger equations in cNLFT which we will report on elsewhere. 

The structure of the paper is the following: 
The concept of 2-graphs is introduced in Sec.~\ref{sec:combinatorics}, followed by the definition of 2-graph contraction in Sec.~\ref{sec:contraction} with a discussion on relevant subtleties concerning the connectedness of the 2-graph's boundary and topology.
In Sec.~\ref{sec:coalgebra} I define the co-algebra and prove the central identity.
Finally, Sec.~\ref{sec:hopfalgebra} provides the general Hopf algebra and explains how to obtain the Connes-Kreimer Hopf algebra for a specific cNLFT 
as the subalgebra generated by the theory's set of divergent 2-graphs.
I close with two examples, the Hopf algebra for the Grosse-Wulkenhaar model and for tensorial field theories.

\newpage

\section{Combinatorial basis: 2-graphs}\label{sec:combinatorics}

For perturbative field theory it is convenient to define graphs in terms of half-edges associated to vertices which are then pairwise combined into edges \cite{Yeats:2008,Yeats:2017,Borinsky:2018}. 
This captures nicely the appearance of self-loops, multi-edges and external legs in Feynman diagrams.

\begin{defin}[graph]\label{1grapha}
A \emph{1-graph}, or simply \emph{graph}, is a tuple $\ogl=(\V,\H,\v,\ei)$ with
\begin{enumerate}[label=(\arabic*),align=left,leftmargin=*, noitemsep]
\item a set of \emph{vertices} $\V$,
\item a set of \emph{half-edges} $\H$,
\item an adjacency map $\v : \H \to \V$ associating half-edges to vertices,
\item[(4a)] an involution on $\H$, that is $\ei : \H \to \H$ such that $\ei\circ\ei = \id$. The resulting pairs of half-edges are called \emph{edges}.
\end{enumerate}
The involution may have fixed points, half-edges paired to themselves. These are understood as \cdef{external edges} (or \cdef{legs}). 
\end{defin}
Alternatively, one can define a graph in terms of an explicit set of edges:
\begin{defin}[graph with edge set]\label{def:1graphb}
A 1-graph is a tuple $(\V,\H,\v,\E)$ with the above properties (1) -- (3) and
\begin{enumerate}[label=(\arabic*),align=left,leftmargin=*]
\item[(4b)] a set of disjoint, two-element subsets of $\H$,
(that is $\E\In\pset{\H}$ such that
$e_{1}\cap e_{2}=\emptyset$ for all $e_{1},e_{2}\in\E$ and $|e|=2$ for all $e\in\E$). 
\end{enumerate}
\end{defin}

The two definitions in terms of either property (4a) or (4b) are equivalent \cite{Borinsky:2018} since an involution (4a) partitions $\H$ into sets of either one or two elements where the latter are in one-to-one correspondence to edge sets (4b).

The simplest way to extend from such notion of graphs to the general class of Feynman diagrams of cNLFT is to introduce another pairing of half-edges, only this time between those at a vertex, to capture the strands.
However, it can happen that also the stranding at a vertex between different edges is multiple, and there can also be self-loops at a single edge.
Thus it is not sufficient to describe the stranding by an involution but it is necessary to explicitly introduce a second layer of (dimension-two) half-edges for the strands:

\begin{defin}[2-graph]\label{def:2grapha}
A \cdef{2-graph} is a tuple $\tgl=(\V,\H,\v,\ei;\S,\h,\vs,\es)$ with the above properties (1) -- (3), (4a) and
\begin{enumerate}[label=(\arabic*),align=left,leftmargin=*,start=5]
\item a set of strand sections $\S$
\item an adjacency map $\h : \S \to \H$ associating them to half-edges,
\item[(7a)] a fixed-point free involution $\vs : \S \to \S$ pairing strand sections at a given vertex, that is
for every $s\in\S$: $\v\circ\h\circ\vs(s) = \v\circ\h(s) $,
\item[(8a)] an involution $\es : \S \to \S$ describing the pairing of strands along edges, that is
\begin{itemize}[noitemsep]
\item $\ei\circ\mu(s) = \mu\circ\es(s)$ for all strand sections $s\in\S$ 
\item every $s\in\S$ is a fixed point of $\es$ iff $\mu(s)$ is a fixed point of $\ei$.
\end{itemize}
\end{enumerate}
\end{defin}

Note that the fulfillment of properties (4a) for $\ei$ and  (8a) for $\es$ depend on each other: There can only be an edge between two half-edges when both have the same number of adjacent strand sections. And if there is an edge, than the adjacent strands have to be paired along that edge.
In that sense $\ei$ and $\es$ together define  \cdef{stranded edges}.
One could collapse the two defining properties into one (in particular $\ei$ is eventually redundant), but since the involutions act on different objects (half-edges and strands) and both are needed for some purpose, it is clearer to separate them.

\begin{defin}[corollae, vertex graph]
It is common to refer to the preimage $\v^{-1}(v)$ of a graph's vertex $v$ as \cdef{corolla}. Its cardinality gives the \cdef{degree} of the vertex $\vd = |\v^{-1}(v)|$.
In a 2-graph $G$, also half edges have a corolla, that is $\h^{-1}(h)$ for a half-edge $h\in\H_G$, as well as a degree $\hd  = |\h^{-1}(h)|$.
Thus, the full \cdef{2-corolla} of a vertex is $(\v\circ\h)^{-1}(v)$ for $v\in\V_G$.

For the entire structure at the vertex $v\in\V$ it is necessary to include the pairing of the strand sections in the 2-corolla. I call this the \cdef{vertex graph}
\[\label{eq:vertexgraph}
\vg=\left( \V_{v},\H_{v},\v_{v},\ei_{v}\right)
 :=\left( \v^{-1}(v), (\v\circ\h)^{-1}(v),\h |_{\H_{v}},\vs |_{\H_{v}} \right)
\]
in which half-edges become vertices and strands become edges.
Note that vertex graphs are not necessarily connected but can have several components.
Such a disconnected vertex is also called a \cdef{multi-trace} vertex.
\end{defin}

\begin{remark}[vertex-graph representation]\label{rem:vertexgraphs}
Since the strand sections $\S$ are paired at individual vertices $v\in\V$ (property (7a)), a 2-graph partitions into its vertex graphs. 
That is, there is a bijective map from 2-graphs to a class of graphs equipped with an additional edge structure%
\[\label{eq:vertexgraphbijection}
\vgb :  (\V,\H,\v,\ei;\S,\h,\vs,\es) \mapsto \big(\{\vg\}_{v\in\V},\ei,\es \big) .
\]
The 2-graph can thus be understood as a gluing of vertex graphs along their half-edge corollae.
I will call this the \cdef{vertex-graph representation} of a 2-graph. 
This is a common way physically relevant classes of 2-graphs are defined in the literature. 
In particular they arise in invariant tensor models, group field theory and spin-foam models where the vertex graphs have been given various colourful names, \eg  ``bubbles''  connected by 0-edges \cite{Bonzom:2012bg}, 
``atoms'' bonded along ``patches'' \cite{Oriti:2015kv}, or ``squid graphs'' glued along their ``squids''   \cite{Kaminski:2010ba}.

It has to be emphasized that it is essential for the bijectivity of the map that its image is defined in terms of 
the set $\{\vg\}_{v\in\V}$ of vertex graphs and not just using
the disjoint union $\bigsqcup_{v\in\V} \vg$.
A single vertex graph $g=\vg$ can have several connected components, $g=\bigsqcup_j g_j$.
Thus, in the disjoint union $\bigsqcup_{v\in\V} \vg$ the information which connected components belong to one vertex in the 2-graph is lost. 
The map to a vertex-graph representation using the disjoint union,
\[\label{eq:vertexgraphprojection}
\vgp :  (\V,\H,\v,\ei;\S,\h,\vs,\es) \mapsto \big(\bigsqcup_{v\in\V} \vg,\ei,\es \big) .
\]
is therefore a projection. 
For the algebraic structures defined below based on contraction and insertion operations it will be crucial that $\vgb$ is bijective and conserves the vertex-belonging information. 
The underlying physical reason is that it is the coupling constants associated to vertices which have to be renormalized.
\end{remark}

\begin{figure}
    \centering
    \begin{tikzpicture}[scale=.7]
    \node [c, label=0:$\cong$] at (-4.2,0) {};
    \node [c, label=0:$\cong$] at (4.2,0) {};
    \begin{scope}[xshift=-7cm]
    \node [v, label=30:1]	(1)	at (-1.5,0) {};
    \node [v, label=150:2, label=-90:3, label=90:4]	(234)	at (0,0) {};
    \node [v, label=-90:5, label=90:6, label=30:7]	(789)at (1.5,0) {};
    \node [c]	(9)	at (2.25,0) {};
    \path
    (1) edge [e] (234)
    (234) edge [e, bend right=80] (789)
    (234) edge [e, bend left=80] (789)
    (789) edge [e] (9);
    \end{scope}
    
    \node [vs,label=80:1]	(1a)	at (-2,.05)	{};
    \node [vs]				(1b)	at (-2,-.05)	{};
    \node [c]				(1c)	at (-3,0)	{};
    \node [c]				(1d)	at (-2.5,0)	{};	
    \node [vc, fit=(1a) (1c)] 	{};
    \node [cs,fit=(1a) (1b)] 	{};    
    \node [vs, label=30:2]	(24)	at (-1,.05)	{};
    \node [vs]				(23)	at (-1,-.05)	{};
    \node [c]				(2c)	at (1,-.05)	{};
    \begin{scope}[rotate=120]  
    \node [vs, label=below:3]	(32)	at (-1,.05)	{};
    \node [vs]				(34)	at (-1,-.05)	{};
    \end{scope}
    \begin{scope}[rotate=-120]  
    \node [vs]				(43)	at (-1,.05)	{};
    \node [vs,label=above:4]	(42)	at (-1,-.05)	{};
    \end{scope}
    \node [vc, fit=(23) (2c)] {};
    \node [cs,fit=(24) (23)] {};
    \node [cs,fit=(32) (34)] {};
    \node [cs,fit=(42) (43)] {};
    \path
    \foreach \i/\j in {24/42,32/23,43/34}{
    (\i) edge [ev,bend right=30]  (\j)
    };    
    \begin{scope}[xshift=2.5cm]
    \node [vs, label=0:7]   (89)	at (1,.05)	{};
    \node [vs]				(87)	at (1,-.05)	{};
    \node [c]				(8c)	at (-1,-.05){};
    \begin{scope}[rotate=120]  
    \node [vs, label=80:6]  (97)	at (1,.05)	{};
    \node [vs]				(98)	at (1,-.05)	{};
    \end{scope}
    \begin{scope}[rotate=-120]  
    \node [vs]				(78)	at (1,.05)	{};
    \node [vs,label=-80:5]  (79)	at (1,-.05)	{};
    \end{scope}
    \node [vc, fit=(89) (8c)] {};
    \node [cs,fit=(98) (97)] {};
    \node [cs,fit=(79) (78)] {};
    \path
    \foreach \i/\j in {89/98,97/79,78/87}{
    (\i) edge [ev,bend left=30]  (\j)
    };
    \node [cs,fit=(89) (87)] {};
    \end{scope}    
    \path
    (1a) edge [ev, bend right=70] (1d)
    (1b) edge [ev, bend left=75]	 (1d)
    (1a) edge [ev] (24)
    (1b) edge [ev] (23)
    (42) edge [ev, bend left=60] (98)
    (43) edge [ev, bend left=60] (97)
    (34) edge [ev, bend right=60] (79)
    (32) edge [ev, bend right=60] (78);
    
    \begin{scope}[xshift=8cm]
    \node [v]	(v1)	at (-2.5,0) {};
    \node [c]	(c1)	at (-3,0) 	{};
    \node [v]	(234)	at (0,0) {};
    \node [v]	(789)at (2.6,0) {};
    \node [vb, label=90:1]	(1)	at (-2,0)	{};
    \node [vb, label=90:2]	(2)	at (-1,0)	{};
    \node [vb, label=-90:3]	(3)	at (.6,-1)	{};
    \node [vb, label=90:4]	(4)	at (.6,1)	{};
    \node [vb, label=90:7]	(9)	at (3.6,0)	{};
    \node [vb, label=90:6]	(8)	at (2,1)	{};
    \node [vb, label=-90:5]	(7)	at (2,-1)	{};
    \path
    (v1) edge [e] (1)
    \foreach \i in {2,3,4}{
    (\i) edge [e] (234)
    }
    \foreach \i in {7,8,9}{
    (\i) edge [e] (789)
    }
    (1) edge [ev, bend right=90] (c1)
    (1) edge [ev, bend left=90] (c1)
    \foreach \i/\j in {2/3,3/4,4/2,7/8,8/9,9/7}{
    (\i) edge [ev] (\j)
    }
    \foreach \i/\j in {1/2,3/7,4/8}{
    (\i) edge [es] (\j)
    };
    \end{scope}        
    \end{tikzpicture}
    \caption{A combinatorial map $(\H,\sigma,\ei)=(\{1,2,3,4,5,6,7\},(1)(234)(576),(12)(35)(46))$ drawn on the plane with counter-clockwise orientation of the vertices (left) and the corresponding \mbox{2-graph} $(\V,\H,\v,\ei;\S,\h,\vs,\es)$ in the stranded representation (ribbon graph, middle) and vertex-graph representation (right; dashed lines represent the pairing $\es$ into stranded edges). 
    For each half-edge $j\in\H$ there are strand sections $s_{ji},s_{jk}\in\S$ according to the cycle $(...ijk...)$ of the permutation $\sigma$, for example here $s_{32}$ and $s_{34}$ adjacent to $h_3$, and $\vs$ pairs each $s_{ij}$ with $s_{ji}$.
    As example of the special case $n<3$, the univalent vertex $v_1$ defined by the cycle $(1)$ in $\sigma$ has two strand sections $s_{11}$ and $s'_{11}$ at the half-edge $h_1$ which are paired with each other by $\vs$.}
    \label{fig:combmap}
\end{figure}
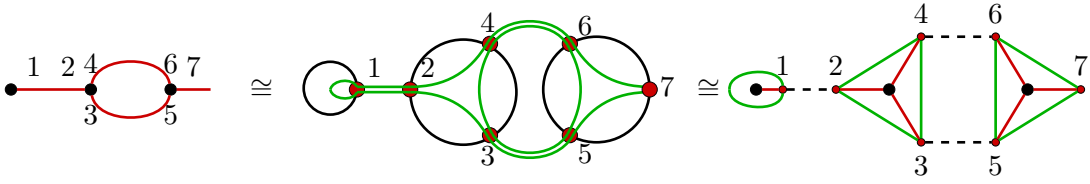

The framework of 2-graphs is very general and covers many examples of field theory.
Even standard local field theory is covered if one considers trivial 2-graphs with $\S=\emptyset$.
The paradigmatic example in mind, though, are theories of tensor fields, in particular 
matrix field theories \cite{Wulkenhaar:2019,Hock:2020tg} (related to noncommutative field theory) with combinatorial maps as diagrams and theories with tensorial interactions of arbitrary rank whose diagrams are coloured graphs \cite{BenGeloun:2014gp,Carrozza:2013uq}:

\begin{example}[combinatorial maps]\label{ex:matrixdiagrams}
Combinatorial maps (also called ``ribbon graphs'' or ``fat graphs'' in physics) are a special example of 2-graphs with all edges of degree $\hd=2$:
A (finite) \emph{combinatorial map} is a triple $(\H,\sigma,\ei)$ with the above properties (1) and (4a)%
\footnote{
In contrast to combinatorics literature (e.g.\,\cite{Eynard:2016ed}), I include boundaries as fixed-points of $\ei$ (external edges) and not in terms of marked edges.
In this way,  diagrams which are disconnected upon removing a disconnected boundary are in fact disconnected maps.
}
and a permutation $\sigma:\H\rightarrow\H$ whose cycles are called vertices.

In this case it is not necessary to define the strand sections explicitly; like the vertices, they are already encoded in the cyclic structure around vertices given by $\sigma$. 
Thus a $\sigma$-cycle $v = (h_{1},h_{2},...,h_{n})$, $n\ge3$, defines the full vertex graph
\[
g_v = (\{h_{1},h_{2},...,h_{n}\},\{s_{1n},s_{12},s_{21},s_{23},...s_{n\,n-1},s_{n1}\},\v_{v},\ei_{v})
\]
where $\v_{v}$ encodes the adjacency of each strand $s_{ij}$ to the half-edge $h_{i}$  and $\ei_{v}$ pairs $s_{ij}$ with $s_{ji}$ iff $\sigma(h_{i})=h_{j}$. 
Also univalent and bivalent vertices (one-cycles and two-cycles) are captured in this way, one only has to distinguish $s_{11}$ and $s'_{11}$ respectively $s_{12},s_{21}$ and $s'_{12}, s'_{21}$.
See Fig. \ref{fig:combmap} for an example.

In the same way, the edge involution $\ei:\H\to\H$ defines already a pairing $\es$ of the edges' strands in terms of the orientation given by $\sigma$. 
Then, the explicit 2-graph is given by the vertex graphs and the bijection $\vgb$,  \eqref{eq:vertexgraphbijection}.
In particular, if all vertices have degree $n\ge3$, faces are simply the cycles of $\sigma\circ\ei$ \cite{Eynard:2016ed}.
\end{example}

\begin{example}[coloured graphs]\label{ex:colouredgraphs}
Another standard example are the Feynman diagrams of rank-$\rk$ tensor theories with invariant interactions \cite{Bonzom:2012bg,Gurau:2016}.
These are $(\rk+1)$-regular, edge-coloured graphs, or \emph{$(\rk+1)$-coloured graphs} for short, that is graphs for which each vertex is adjacent to exactly one edge decorated with ``colour'' $c=0,1,...,\rk$ each.
The interpretation as a Feynman diagram, and thus the bijection to 2-graphs is the following (see also Fig. \ref{fig:colouredgraph}):

For a $(\rk+1)$-coloured graph, the connected components of the $\rk$-coloured graph obtained by deleting all 0-edges define a set of vertex graphs; that is, the coloured edges are bijective to coloured strand sections in the 2-graph with pairing $\vs$ at each vertex according to this graph structure.
Clearly, all half edges $h$ are $\rk$-valent, $\hd=\rk$.
Then the $0$-edges define then stranded edges whereby a unique pairing $\es$ of the strand sections follows from the condition that only pairings of strand sections of the same colour are allowed. 
This means that in the result, the \emph{$\rk$-coloured 2-graph}, the involution $\es$ is redundant. 

This bijection maps only to coloured 2-graphs with connected vertices, that is all vertices have connected vertex graphs. In fact, the bijection is exactly the map $\vgp$ discussed in Rem.\,\ref{rem:vertexgraphs} which is bijective only upon restriction to 2-graphs with connected vertices. To extend to coloured 2-graphs with disconnected vertices the bijection $\vgb$ is needed. 
For this the coloured graphs lack the additional information which subsets of connected components form 2-graph vertices after 0-edge deletion.
\end{example}

\begin{figure}
    \centering
    \begin{tikzpicture}[scale=.6]
    \node [c, label=0:$\cong$] at (6.6,0) {};
    \node [c]	(a1)	at (-1.7,-1.7)	{};
    \node [c]	(c1)	at (-1.7,1.7)	{};
    \node [v]	(w1)	at (-1,-1)	{};
    \node [v]	(b1)	at (-1,1)	{};
    \node [v]	(w2)	at (1,1)	{};
    \node [v]	(b2)	at (1,-1)	{};
    \node [c]	(a4)	at (5.7,1.7)	{};
    \node [c]	(c4)	at (5.7,-1.7)	{};
    \node [v]	(w3)	at (3,-1)	{};
    \node [v]	(b3)	at (3,1)	{};
    \node [v]	(w4)	at (5,1)	{};
    \node [v]	(b4)	at (5,-1)	{};
    \path
\foreach \i/\j in {1/2}{
   	(w\i) edge   (b\j)
	(b\i) edge  node [c, label=above:$c_1$] {} (w\j)
	}
\foreach \i/\j in {3/4}{
   	(w\i) edge   (b\j)
	(b\i) edge node [c, label=above:$c_2$] {} (w\j)
	}
    \foreach \i in {1,2,3,4}{
	(w\i) edge 				node 	{}	(b\i)
	(w\i)	edge [bend left=30]	node 	{}	(b\i)
	(w\i)	edge [bend right=30]	node 	{}	(b\i)
	}
    \foreach \i/\j in {a1/w1,c1/b1,a4/w4,c4/b4,w2/b3,w3/b2}{
   	(\i) edge [es] (\j)
	};

    \begin{scope}[xshift=10cm]
    \node [v]	(1)	at (0,0)	{};
    \node [v]	(2)	at (4,0)	{};
    \node [vb]	(w1)	at (-1,-1)	{};
    \node [vb]	(b1)	at (-1,1)	{};
    \node [vb]	(w2)	at (1,1)	{};
    \node [vb]	(b2)	at (1,-1)	{};
    \node [vb]	(w3)	at (3,-1)	{};
    \node [vb]	(b3)	at (3,1)	{};
    \node [vb]	(w4)	at (5,1)	{};
    \node [vb]	(b4)	at (5,-1)	{};
    \path
\foreach \i/\j in {1/2}{
   	(w\i) edge [ev]   (b\j)
	(b\i) edge [ev]  node [c, label=above:$c_1$] {} (w\j)
	}
\foreach \i/\j in {3/4}{
   	(w\i) edge [ev]   (b\j)
	(b\i) edge [ev]  node [c, label=above:$c_2$] {} (w\j)
	}
    \foreach \i in {1,2,3,4}{
	(w\i) edge 	[ev]				(b\i)
	(w\i)	edge [ev,bend left=30]	(b\i)
	(w\i)	edge [ev,bend right=30]	(b\i)
	}
    \foreach \i in {w1,b1,w2,b2}{
   	(\i) edge [e] (1)
	}
    \foreach \i in {w3,b3,w4,b4}{
   	(\i) edge [e] (2)
	}	
    \foreach \i/\j in {w2/b3,w3/b2}{
   	(\i) edge [es] (\j)
	};
    \end{scope}
    \end{tikzpicture}
    \caption{Example of a (4+1)-coloured graph (left) and its corresponding 4-coloured 2-graph (in its vertex-graph representation, right). 
    In the coloured graph, dashed lines represent colour-$0$ edges, internal and external.
    The internal ones become the edges in the 2-graph whose strand structure is here completely determined by the strands of colour $c_1,c_2\in\{1,2,3,4\}$.
    }
    \label{fig:colouredgraph}
\end{figure}
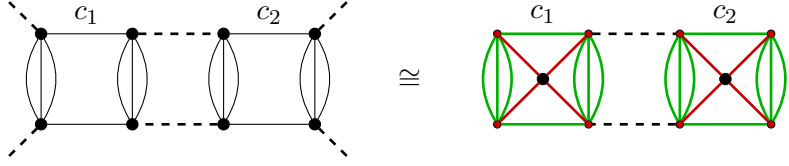

For some purposes it is useful to define the pairings $\vs$ and $\es$ explicitly in terms of subsets of $2^{S}$ like the edges in Def.\,\ref{def:1graphb}:
\begin{defin}[2-graph with edge set]\label{def:2graphb}
A \cdef{2-graph} is a tuple $(\V,\H,\v,\E;\S,\h,\VS,\ES)$ with the above properties (1) -- (4b) -- (6) and
\begin{enumerate}[label=(\arabic*),align=left,leftmargin=*,start=7]
\item[(7b)] a complete partition of $\S$ into disjoint subsets of two elements adjacent to the same vertex, that is $\VS\In\pset{\S}$ such that
\begin{itemize}[noitemsep]
\item $a\cap b=\emptyset$ for all $a,b\in\VS$,
and $|a|=2$ for all $a\in\VS$,
\item for every vertex $v\in\V$ and every $s_{1}\in(\nu\circ\mu)^{-1}(v)$ there is an $s_{2}\in(\nu\circ\mu)^{-1}(v)$ such that $\{s_{1},s_{2}\}\in\VS$.
\end{itemize}
\item[(8b)] a set of disjoint, two-element subsets of $\S$ compatible with the edges $\E$, that is $\ES\In\pset{\S}$ such that
\begin{itemize}[noitemsep]
\item $a\cap b=\emptyset$ for all $a,b\in\ES$
and $|a|=2$ for all $a\in\ES$,
\item for every $e=\{h_{1},h_{2}\}\in\E$ and every $s_{1}\in\mu^{-1}(h_{1})$ there is an $s_{2}\in\mu^{-1}(h_{2})$ such that $\{s_{1},s_{2}\}\in\ES$.
\item An $h\in\H$ is not contained in any edge $e\in\E$ iff all $s\in\mu^{-1}(h)$ are not contained in any edge strand $a\in\ES$.
\end{itemize}
\end{enumerate}
\end{defin}

\begin{proposition}
The two 2-graph definitions, Def.\,\ref{def:2grapha} and Def.\,\ref{def:2graphb}, are equivalent.
\end{proposition}
\begin{proof}
The fixed-point free involution $\vs$ partitions $\S$ into disjoint 2-element sets $\VS\In 2^\S$.
Property (7a) says that such pairs are adjacent to the same vertex $v\in\V$. Thus, for every $s_1\in\S$ adjacent to $v\in\V$ there is also the unique $s_2=\vs(s_1)\ne s_1$ adjacent to $v$ which gives (7b), and (7a) follows from (7b) in the same way.

The same argument applies to the adjacency of pairs $\{s,\es(s)\}\in\ES$ to an edge $\{h,\ei(h)\}\in\E$.
The equivalence of fixed points of $\ei$ and $\es$ is equivalent to the equivalence of half-edges $h\in\H$ and their adjacent strands $s\in\mu^{-1}(h)$ not occurring in the power sets $\E$ and $\ES$ respectively. Together this shows that (8a) and (8b) are equivalent.
\end{proof}

\begin{defin}[external edges and strand sections]
Half-edges which are not part of any edge are called \cdef{external},
$\Hx:=\H\setminus\bigcup_{e\in\E}e$.
Equivalently, their associated strand sections are called \cdef{external strands}, $\Sx:=\S\setminus\bigcup_{s\in\ES}s$.
For some purpose it is meaningful to define \cdef{external edges} as one-element sets  $\Ex:=\{\{h\}|h\in\Hx\}$ compatible with the actual (internal) edges.
\end{defin}


\begin{defin}[faces]\label{def:faces}
Let $f=(s_{1},s_{2},...,s_{2n})$ be an $2n$-tuple of distinct strand sections $s_{i}\in\S$, $i=1,2,...,2n$, which are each mapped to the following one by $\vs$ and $\es$ alternating:
\[
s_{1}\overset{\vs}{\mapsto}s_{2}\overset{\es}{\mapsto}s_{3}\overset{\vs}{\mapsto}...\overset{\vs}{\mapsto}s_{2n}.
\]
Then $f$ is an \cdef{external face} iff $s_{1}$ and $s_{2n}$ are fixed points of $\es$, \ie they are external.

If $\es(s_{2n}) = s_{1}$ there is an equivalence relation of cyclic permutations $s_{i}\mapsto s_{i+2}$. The equivalence class $[(s_{1},s_{2},...,s_{2n})]$ is called an \cdef{internal face}.

In the following, $\Fx$ denotes the set of external faces, $\Fi$ of internal faces and $\F=\Fx\cup\Fi$.
As faces are complete chains of strand sections, they are sometimes simply called strands. This is why, to avoid confusion, I refer to the elements $s\in\S$ as strand \emph{sections}. 
\end{defin}

The reason to call the chains of strand sections \emph{faces}, and the diagrams themselves \emph{2-graphs} is that they are indeed 2-dimensional complexes (in the sense of Reidemeister \cite{Reidemeister:1938vf}; the reason to use the old combinatorial definition is simply that 2-graphs are also purely combinatorial objects \cite{Thurigen:2015}). 
\begin{proposition}\label{prop:complex}
2-graphs are pure complexes. If a 2-graph has no trivial (degree-0) vertices nor trivial half-edges it is furthermore a 2-dimensional complex.
\end{proposition}

\begin{proof}
Recall that according to \cite{Reidemeister:1938vf} a complex $(\cp,\dm,\le)$ is a set $\cp$ of cells $c$ with a dimension map  $\dm:\cp\to\N$   and a partial ordering $\le$ 
that obeys the property (CP): If $c>c''$ and $\dm(c)-\dm(c'')>1$, then there is a cell $c'$ such that $c>c'>c''$.

For a 2-graph $G$, define the complex
\[
\cp_{G}:=\V_{G}\cup(\E_{G}\cup\Ex_{G})\cup\F_{G}
\]
and assign the dimension $0,1,2$ to vertices, edges and faces respectively. The adjacency maps $\v$ and $\h$ induce a partial ordering:
For $v\in\V, e\in\E$ and $f\in\F$: $v<e$ iff $h\in e$ such that $\v(h)=v$; $e<f$ iff there is an $h\in e$ and $s\in f$ (independent of a chosen representative of $f$) such that $\h(s)=h$ ; and $v<f$ iff there is an $e\in\E$ such that $v<e<f$.
Then (CP) holds by definition.

A complex is pure if each cell of non-zero dimension bounds a 0-cell \cite{Reidemeister:1938vf}. By definition of the bounding relation via adjacency maps $\v$ and $\h$ this holds.

A complex is $n$-dimensional iff $n$ is the maximal dimension of cells and for each cell there is a bounding $n$-cell \cite{Reidemeister:1938vf}. For $\cp_{G}$ the maximal dimension is $n=2$ by definition.
However, as $\v$ and $\h$ are not necessarily surjective, there might be vertices or half-edges not bounded by a face. Thus $\cp_{G}$ is 2-dimensional only if both are surjective, that is all vertices and half-edges are not 0-valent.
\end{proof}


\begin{remark}[$\std$-dimensional complexes from edge-coloured graphs]
\label{rem:manifolds}
With some additional structure, 2-graphs can also be bijective to complexes of higher dimension, and even to discrete (pseudo) manifolds.
In fact, bipartite $(\rk+1)$-coloured graphs (Ex.~\ref{ex:colouredgraphs}) are dual to $\rk$-dimensional abstract simplicial pseudo manifolds. This works since the connected components upon deleting edges of $p$ colours define $(p-1)$-dimensional cells~\cite{Gurau:2011dw,Gurau:2010iu}.

\end{remark}


\section{Contraction and boundary}\label{sec:contraction}

The central operation for the coproduct on 2-graphs, and thus for the Hopf algebra, is contraction of a subgraph.
In effect, the components of the subgraph are substituted by single vertices with the subgraph's external structure, usually called the \emph{residue}.
For a \mbox{2-graph} the residue is bijective to a boundary 1-graph.
Since such a boundary 1-graph can be disconnected even for a connected 2-graph
it is crucial to keep track of which boundary (1-graph) connected component belongs to which bulk (2-graph) connected component.

\begin{defin}[subgraph]
For a 2-graph $\tgl$, a \cdef{subgraph} $\sgl$ is a 2-graph which is only different from $\tgl$ in having $\E_\sgl\In\E_\tgl$ and $\ES_\sgl\In\ES_\tgl$. Then one writes $\sgl\sgr \tgl$.
\end{defin}
Note that $\E_\sgl$ and $\ES_\sgl$ must still be compatible due to the 2-graph property (8b). They form indeed stranded edges.


\begin{defin}[contraction]\label{def:contraction}
For 2-graphs $\sgl\sgr \tgl$ the contraction of $\sgl$ in $\tgl$ is defined by shrinking all stranded edges which belong to $\sgl$. That is, the contracted graph $\tgl/\sgl$ consists of
\begin{itemize}[leftmargin=*, noitemsep]
\item $\V_{\tgl/\sgl} = \cc_{\sgl}$ the set of connected components of $\sgl$, that is each connected component in $\sgl$ is shrunken to a single vertex,
\item $\H_{\tgl/\sgl} = \Hx_{\sgl}$, $\S_{\tgl/\sgl} = \Sx_{\sgl}$, only half-edges and strand sections which are external in $\sgl$ remain in $\tgl/\sgl$,
\item $\v_{\tgl/\sgl} = \pi_\sgl\circ\v_{\sgl}\big|_{\H_{\tgl/\sgl}}$ where $\pi_\sgl:\V_{\sgl}\to\cc_{\sgl}$ is the projection of vertices on their connected component,
\item $\h_{\tgl/\sgl} = \h_{\sgl} \big|_{\S_{\tgl/\sgl}}$, simply the restriction to the remaining strand sections,
\item $\E_{\tgl/\sgl} = \E_{\tgl}\setminus\E_{\sgl}$, $\ES_{\tgl/\sgl} = \ES_{\tgl}\setminus\ES_{\sgl}$,
stranded edges of $\sgl$ are deleted and
\item $\VS_{\tgl/\sgl} = \left\{\{s_{1},s_{2n}\} | (s_{1}...s_{2n})\in\Fx\right\}$, external faces are shrunken to the strands at the new contracted vertices.
\end{itemize}
Note that internal faces are deleted completely since $\S_{\tgl/\sgl} = \Sx_{\sgl}$. 
\end{defin}


Take as an example the coloured 2-graph  of Fig.\,\ref{fig:colouredgraph}, with explicit labelling of half-edges
\[\label{eq:r4fishgraph}
G = \cfishl \, .
\]
It has two edges, and thus $2^2=4$ different subgraphs
\[\label{eq:subgraphexample}
H_{0} =  \cfishc \,,\;  H_1=\cfisha \,,\;  H_2=\cfishb \,,\;  H_3=G=\cfish 
\]
For contraction one has to distinguish the cases $c_1=c_2=c$ and $c_1\ne c_2$. In the first case,
\[\label{eq:contractionexample1}
G/H_0 = \cfishlc \,,\; G/H_1=\cvsla \,,\; G/H_2=\cvslb \,,\; G/G=\cvfl
\]
while for $c_1\ne c_2$ the contractions are
\[\label{eq:contractionexample2}
G/H_0 = \cfishl \,,\; G/H_1=\cvslc \,,\; G/H_2=\cvsld \,,\; G/G=\cvftl\,.
\]
The last contraction $G/G$ is an example of a multi-trace vertex, \ie a vertex with disconnected vertex graph.

\

In field theory, the structure of graphs is relevant up to relabelling their components.
In this sense, $H_1$ and $H_2$ in \eqref{eq:subgraphexample} are equivalent and so are  $G/H_1$ and $G/H_2$.

\begin{defin}[iso/automorphism]
An \cdef{isomorphism} $j$ between 2-graphs $\tgl_1$ and $\tgl_2$ is a triple of bijections $j=(\jv,\jh,\js)$ where $\jv:\V_{\tgl_1}\to\V_{\tgl_2}$, $\jh:\H_{\tgl_1}\to\H_{\tgl_2}$ and  $\js:\S_{\tgl_1}\to\S_{\tgl_2}$ such that 
\begin{itemize}[leftmargin=*,noitemsep]
    \item $\v_{\tgl_2} = \jv\circ\v_{\tgl_1}\circ\jh^{-1}$ and $\h_{\tgl_2} = \jh\circ\h_{\tgl_1}\circ\js^{-1}$,
    \item $\ei_{\tgl_2} = \jh\circ\ei_{\tgl_1}\circ\jh^{-1}$,
    \item $\vs_{\tgl_2} = \js\circ\vs_{\tgl_1}\circ\js^{-1}$ and $\es_{\tgl_2} = \js\circ\es_{\tgl_1}\circ\js^{-1}$.
\end{itemize}
A 2-graph \cdef{automorphism} is an isomorphism from a 2-graph $G$ to itself.
\end{defin}
\begin{defin}[unlabelled 2-graphs]
Two 2-graphs $\tgl_1$ and $\tgl_2$ are equivalent upon relabelling,
$\tgl_1 \tsim \tgl_2$ ,
iff there is a 2-graph isomorphism $\tgl_1\to\tgl_2$. Such an equivalence class is an \emph{unlabelled 2-graph},
$\tg=[\tgl_1]_{\tsim} =[\tgl_2]_{\tsim}$.
Let $\TG$ denote the set of all unlabelled 2-graphs and $\OG$ the set of unlabelled 1-graphs. 
\end{defin}
Thereby I 
use the convention to denote unlabelled objects by Greek letters while labelled ones by Roman letters.
Capital letters refer to 2-graphs and small letters to 1-graphs.

\begin{defin}[residue]
In analogy to usual Feynman graphs, one refers to the class of 2-graphs without any stranded edge as \cdef{residues} $\Res^{*}\In\TG$ and denotes the set of those with a single vertex as $\Res\In\Res^{*}$.
Shrinking all stranded edges of a 2-graph results in a residue. 
This is commonly defined as the \cdef{residue map} 
\[
\res:\TG\to\Res^{*},\quad \tg\mapsto\tg/\tg.
\]
Furthermore, for every 2-graph $\tg\in\TG$ there is a trivial subgraph $\sg_0\sgr\tg$ without stranded edges, called the \cdef{skeleton} (not to be confused with the $p$-skeleton of a cell complex):
\[
\skl:\TG\to\Res^{*},\quad \tg \mapsto \sg_0 \, .
\]
This can be used to define a subset of 2-graphs with a given type of vertices $\Vtyper\In\Res^*$ as
\[\label{eq:2graphsVtyperesidue}
\TG(\Vtyper) := \left\{\tg\in\TG \, \middle| \,  \skl(\tg)\in\Vtyper\In\Res^* \right\}
\, .\]
\end{defin}

\begin{remark}[boundaries]
The residue is very similar to the usual notion of a boundary for 2-dimensional complexes, with subtle but crucial distinctions.
Indeed, the union of vertex graphs \eqref{eq:vertexgraph} of the residue of a 2-graph $\tg\in\TG$ is its boundary. 
Thus it is meaningful to define the boundary map
\[
\br: \TG\to\OG, \quad  
\tg\mapsto \br\tg := \vgp(\res(\tg))
\]
where by slight abuse of notation the image of $\vgp$ is simply taken as the set of 1-graphs~$\OG$ since there are no stranded edges in the image of $\res$.
It is straightforward to check that this definition of boundary is equivalent to the notion of boundary for a complex according to Prop.~\ref{prop:complex}, in particular in the example of $\rk$-coloured graphs (Ex.\,\ref{ex:colouredgraphs}) to the notion of boundary of $\rk$-dimensional simplicial pseudo manifolds (Rem.\,\ref{rem:manifolds}).

However, for renormalization in field theory it is necessary to distinguish the boundaries of the different connected components of the 2-graph $\tg=\bigsqcup_{i\in I}\tg_i$.
To appropriately take this into account another boundary map is needed,
\[
\brt: \TG\to\pmset{\OG},\quad 
\tg=\bigsqcup_{i\in I}\tg_i\mapsto \brt\tg := \{\br\tg_i\}_{i\in I} \, ,
\]
where here $\pmset{\OG}$ denotes the power set of all multisets of $\OG$ since any $\og\in\OG$ may appear multiple times in $\brt\tg$.
In fact, this boundary map contains exactly the same information as the residue.
That is, the two are bijective by changing to the vertex-graph representation, 
$
\brt\tg 
= \vgb(\res(\tg_i))
$, 
since each connected component $\tg_i$ maps to one vertex under the residue map.
Thus, the difference between $\brt$ and $\br$ mirrors the one between $\vgb$ and $\vgp$ discussed in Rem.\,\ref{rem:vertexgraphs}.

Similarly one can distinguish two types of skeletons in the vertex-graph picture,
\begin{align}
\vgs:\TG\to \OG, & \quad \tg \mapsto \vgs\tg := \vgp(\skl(\tg))=\bigsqcup_{v\in\V_\tg}\og_v \quad\textrm{and}\\ 
\vgt:\TG\to \pmset{\OG}, & \quad \tg \mapsto \vgt\tg := \vgb(\skl(\tg))=\{\og_v\}_{v\in\V_\tg} \, .
\end{align}
This allows for a definition of 2-graphs with specific vertex types in terms of their vertex graphs $\mathbf{V}\In\OG$ as
\[\label{eq:2graphsVtypegraphs}
\TG(\Vtypeg) := \left\{\tg\in\TG \, \middle| \,  \vgt\tg\in\pmset{\mathbf{V}} \right\}
\]
which is equivalent to the definition \eqref{eq:2graphsVtyperesidue} in the sense that for $\Vtyper\In\Res^*$ it holds that $\TG(\Vtypeg)=\TG(\Vtyper)$ iff $\pmset{\Vtypeg}=\brt\Vtyper$.

The important point here is that it makes a difference to compare the external graph structure on the level of vertex graphs or the single graph given by their disjoint union
Take for example the unlabelled 2-graph $\tg=\cfish$ of Fig.\,\ref{fig:colouredgraph} or \eqref{eq:r4fishgraph}.
If the distinguished colours $c_1\ne c_2$ are not the same, the residue $\res(\tg)=\cvft$ is a multi-trace vertex. Thus for several such components, \eg $\tg\sqcup\tg$, there is a difference between
\[\label{eq:boundaryexample}
\br(\tg\sqcup\tg) = \cvtvg \sqcup \cvtvg \sqcup \cvtvg \sqcup \cvtvg 
\quad \textrm{ and } \quad 
\brt(\tg\sqcup\tg) = \left\{\cvtvg\sqcup\cvtvg\,,\,\cvtvg\sqcup\cvtvg \right\} 
\]
as the boundary $\br$ misses the information which boundary component belongs to which 2-graph component.
\end{remark}

\begin{remark}[topology change]
When one can associate manifolds with {2-graphs}, contraction may lead to change of topology.
The boundary of a connected component $\sg_i$ of a subgraph $\sg\sgr\tg$ might be disconnected, that is a 1-graph with several connected components $\br\sg_i=\bigsqcup_j\og_j$, as exemplified in \eqref{eq:boundaryexample}.
Contraction of such $\sg_i$ leads to multi-trace vertex in $\tg/\sg$.
From the field-theory perspective alone, it is not obvious how to understand the topology of such 2-graphs in cases where they relate to pseudo manifolds (\eg surfaces (Ex.\,\ref{ex:matrixdiagrams}) or higher dimensional (Ex.\,\ref{ex:colouredgraphs})):
\begin{enumerate}[label=(\arabic*.),leftmargin=*]
\item One could focus simply on the strands (propagation of degrees of freedom) and ignore the fact that the disconnected parts are adjacent to the same vertex; this would lead to topology change since the contracted component would be separated.
\item Topology change upon contraction could be avoided by simply assigning to the multi-trace vertex with vertex graph $\br\sg_i=\bigsqcup_j\og_j$ the topology of $\sg_i$.
But this might not be unique as there could be various 2-graphs $\sg_i$ related to different topologies.
\item In matrix models such multi-trace vertices are interpreted as singular points.%
\footnote{
In particular, one finds that there is a region in the phase diagram of multi-trace matrix models \cite{Das:1990gp,Korchemsky:1992uj} with the critical behaviour of the continuum random tree \cite{Aldous:1991ja}. This regime is dominated by ``cactus'' geometries, that is surfaces connected along single points thus leading to a tree-like fractal structure.
}
This suggests to understand multi-trace vertices as points at which the complex is only one-dimensionally connected.
\end{enumerate}
In fact, in some examples of surfaces, the last possibility is in agreement with the Euler characteristic.
Consider for example the following contraction of combinatorial maps
\[
\tg/\sg \cong \left.\maptorus \middle/ \mapd \right. = \mapcon\, .
\]
While $\tg$ has torus topology, $\sg$ is a cylinder. 
The contraction of $\sg$ in $\tg$ leads to a double-trace vertex whose strands are disconnected. 
Considering only the strands but not the vertex (1.), $\tg/\sg$ is a sphere different from the original torus (2.). 
On the other hand, taking the vertex as a singular point (3.) gives a pinched torus and in fact a naive calculation of Euler characteristic gives
\[
\chi(\tg/\sg) = \nv - \nei + \nf = 3 - 4 + 2 = 1
\]
as expected for a pinched torus. It is an interesting question whether these topological considerations can be generalized and would apply for example to the $\rk$-dimensional pseudo manifolds of $\rk$-coloured graphs (Ex.\,\ref{ex:colouredgraphs}).
\end{remark}




\section{The 2-graph coalgebra and the central identity}\label{sec:coalgebra}

Algebraic structure can be defined on 2-graphs in the same way as for 1-graphs \cite{Borinsky:2018}. 
With the proper definition of a 2-graph (Def.~\ref{def:2grapha}) and of contraction of subgraphs  (Def.~\ref{def:contraction}) a product, a coproduct as well as an antipode can be defined in a standard way promoting the set of 2-graphs $\TG$ to an algebra, coalgebra and Hopf algebra respectively.
The crucial difference to 1-graphs lies in the combinatorial details as I will exemplify with the central identity for the action of the coproduct on 
power series over 2-graphs in Theorem \ref{th:centralidentity}.

\begin{defin}[2-graph algebra]\label{def:2graphalgebra}
Let $\btg:=\langle\TG\rangle$ be the $\Q$-algebra generated by all 2-graphs $\tg\in\TG$ with multiplication in terms of the disjoint union defined on the generators,
\[
m:\btg\otimes\btg\to\btg \quad , \quad \tg_1\otimes\tg_2\mapsto\tg_1\sqcup\tg_2 \, .
\]
This is clearly associative and commutative and the empty 2-graph $\one$ is the neutral element.
With a unit, that is a linear map $u:\mathbb{Q}\to\btg,q\mapsto q\one$ this is a unital commutative algebra.

The set $\TG(\Vtypeg)$ of 2-graphs  with vertex graphs of restricted types $\Vtypeg\In\OG$, \eqref{eq:2graphsVtypegraphs}, generates a subalgebra $\btg_\Vtypeg:=\langle \TG(\Vtypeg) \rangle$ since by definition $m(\btg_\Vtypeg\otimes\btg_\Vtypeg)\In\btg_\Vtypeg$.
\end{defin}

As common, there is a {coproduct} $\cop$ on the algebra of 2-graphs defined on their generators as a sum over all subgraphs tensored with their contraction.
For example, the coloured 2-graph of \eqref{eq:r4fishgraph} with distinguished colours $c_1\ne c_2$ has
\begin{align}
   \cop\left(\cfish \right) =&\quad \cfishc \otimes \cfish + \cfisha \otimes \cvsd \nonumber\\ &+ \cfish \otimes \cvft 
\end{align}
In this way one obtains a coalgebra and furthermore a bialgebra:

\begin{proposition}[2-graph coalgebra]
Equipping $\btg$ with the \cdef{coproduct}, a linear map defined on its generators as 
\[
\cop: \btg \to \btg \otimes \btg , \quad  \tg \mapsto \sum_{\sg\sgr\tg} \sg \otimes \tg/\sg
\, ,\]
and the counit, a projector defined as
\[
\cou: \btg \to \Q , \quad  \tg \mapsto 
\begin{cases}
1 \textrm{ if }  \tg \in \Res^{*} \\
0 \textrm{ else }
\end{cases} \, ,
\]
$\btg$ is an associative counital coalgebra.
Since $m$ and $\cop$ are compatible in the sense that $\cop$ is an algebra homomorphism and $m$ is a coalgebra homomorphism, $\btg$ is furthermore a (unital and counital ) bialgebra.
\end{proposition}

\begin{proof}
For an associative coalgebra one has to show that 
\[
(\cop \otimes \id)\circ\cop = (\id \otimes \cop) \circ \cop 
\quad \textrm{ and } \quad
(\cou \otimes \id)\circ\cop = (\id \otimes \cou) \circ \cop = \id \, .
\]
The proof is the same as for the 1-graph coalgebra (e.g. \cite{Borinsky:2018} Prop.~5.2.1) and for the bialgebra property \cite{Borinsky:2018}, Prop.~5.2.2.
\end{proof}

In field theory, one is often interested in the subalgebra $\btg_\Vtypeg=\langle \TG(\Vtypeg) \rangle$ of 2-graphs with vertex graphs of specific type $\Vtypeg\in\OG$.
But this is not a subcoalgebra since contractions of 2-graphs in $\btg_\Vtypeg$ might lead to 2-graphs with other vertices, as in the example \eqref{eq:contractionexample1} and \eqref{eq:contractionexample2}.
Thus, in general $\btg_\Vtypeg\in\btg$ is only a right co-ideal, \ie $\cop(\btg_\Vtypeg)\In \btg_\Vtypeg\otimes\btg$.
To upgrade it to a subcoalgebra it is thus necessary to extend the set of 2-graphs by all possible contractions:

\begin{defin}[Contraction closure and subcoalgebras]\label{def:closure}
Given a set $\PG\In\TG$, a set $\KG\In\TG$ is called $\PG$-\cdef{contraction closed} iff for all $\sg\sgr\tg\in\KG$ with $\sg\in\PG$ also $\tg/\sg\in\KG$.

The $\PG$-\cdef{contraction closure} of $\KG\In\TG$ is the extension of $\KG$ by all such contractions,
\[
{}^\PG\overline{\KG} := \left\{\tg=\tg'/\sg \middle| \sg\sgr\tg'\in\KG, \sg\in\PG\right\}
\, .\]
For $\PG=\TG$ one calls the $\TG$-contraction closure of $\KG$ simply the \cdef{contraction closure} and writes $\overline{\KG}$.
\end{defin}

\begin{proposition}[2-graph subbialgebra]\label{prop:subbialgebra}
For a set of 2-graphs $\TG(\Vtypeg)$ of given vertex types $\Vtypeg\In\OG$, the algebra $\langle \overline{\TG(\Vtypeg)}\rangle$ is a subbialgebra of $\btg$.
\end{proposition}
\begin{proof}
Let $\tg\in\overline{\TG(\Vtypeg)}$ be a 2-graph in the contraction closure of 2-graphs of vertex types~$\Vtypeg\In\OG$. 
By definition this means that there are 2-graphs $\tilde{\sg}\sgr\ttg\in\TG(\Vtypeg)$ such that $\tg=\ttg/\tilde{\sg}$. In particular, it might be the case that there are connected components $\ttg_i$ and $\tilde{\sg}_j=\ttg_i$ such that $\ttg_i/\ttg_i=\res(\ttg_i)$ is a connected component of $\tg$.
Thus the vertices of $\tg$ can be of any type of residues (respectively boundaries) of $\TG(\Vtypeg)$, that is $\skl(\tg)\in\res(\TG(\Vtypeg))$.
Since there are no restrictions on edges, this completely characterizes 2-graphs in the closure $\overline{\TG(\Vtypeg)}$.

A subgraph $\sg\sgr\tg$ has by definition $\skl(\sg)=\skl(\tg)$.
Thus, it is also in $\overline{\TG(\Vtypeg)}$.
The contraction $\tg/\sg$ is in the closure by definition and thus $\cop\tg\in\overline{\TG(\Vtypeg)}\otimes\overline{\TG(\Vtypeg)}$, which proves that  $\overline{\TG(\Vtypeg)}$ is a subcoalgebra.
As $\overline{\TG(\Vtypeg)}$ is also a subalgebra, it is a subbialgebra.
\end{proof}

\begin{example}[Bialgebra of maps and coloured 2-graphs]\label{ex:bialgebra}
Both the example of combinatorial maps (Ex.\,\ref{ex:matrixdiagrams}) and of $\rk$-coloured graphs (Ex.\,\ref{ex:colouredgraphs}) are classes of 2-graphs which are characterized by a fixed half-edge degree, $\hd = 2$ respectively $\hd=\rk$, as well as a specific edge stranding $\vs$ induced by orientation respectively colour structure.
The first property, $\hd=\rk$, defines a set of vertex graphs $\Vtypeg_\rk\in\OG$.
In this case, since the residue (boundary) of any 2-graph with such vertex type is again of this type, \ie has fixed $\hd=\rk$, the set $\TG(\Vtypeg_\rk)$ is already contraction closed.
The second property of specific edge stranding is also preserved under contractions.
Thus the set of combinatorial maps and the set of $\rk$-coloured 2-graphs each generate a subbialgebra in the bialgebra of all 2-graphs~$\btg$.
In field theory such sets of diagrams are called ``theory space'' exactly because of this closure property.
\end{example}

The quantity of interest in field theory is the perturbative expansion of Green's functions which is a formal power series labelled by Feynman diagrams.
The corresponding combinatorial object therefore is the formal series
\[
X := \sum_{\tg\in\TG} \frac{\tg}{|\aut\tg|}\, ,
\quad \textrm{ or } \quad 
X^\og = \sum_{\tg\in\TGog} \frac{\tg}{|\aut\tg|}, 
\]
that is its restriction to a sum over $\TGog$, \emph{connected} 2-graphs with a given external boundary structure $\br\tg=\og\in\OG$.
The crucial operation for the renormalization of Green's functions is the coproduct on the series $X^\og$ which is induced by the general case $\cop X$~\cite{Kreimer:2006gg, Yeats:2008, vanSuijlekom:2007jv, Borinsky:2014eg}: 

\begin{theorem}[central identity]\label{th:centralidentity}
The coproduct of the weighted series over all 2-graphs is
\[\boxd{
\cop X = \sum_{\tg\in\TG} \Big( \prod_{v\in V_\tg} |\aut\og_v|\, X^{\og_v}\Big) \otimes \frac{\tg}{|\aut\tg|} \, .
}\]
\end{theorem}

The same formula is also true for more restricted formal power series where the sum runs over a subset $\KG\In\TG$ such as $X^\og$, the sum over 2-graphs with specific boundary $\og$, or the restriction to bridgeless (one-particle irreducible) 2-graphs.
The necessary property for such a subset $\KG$ is to be contraction closed (\cite{Borinsky:2018} Corollary 5.3.1).

To prove the theorem one needs the inverse operation to contraction which is insertion.
However, one has to be careful with the definition since for 2-graphs not any kind of insertion is dual to a contraction.
To insert a labelled 2-graph $\tgla$ into another $\tglb$ appropriately, one replaces all vertices of $\tglb$ with connected components of $\tgla$.
For this it is necessary that the boundaries of the components of $\tgla$ agree with (are isomorphic to) the vertex graphs of $\tglb$, 
that is $\brt \tgla \osim \vgt\tglb$, 
and to specify how edges and strands of $\tglb$ connect to $\brt\tgla$:
\begin{defin}[insertion] 
Let $\tgla$ and $\tglb$ be 2-graphs 
and  $i=(i_{\H},i_{\S}):\br \tgla\to \vgs\tglb$ a 1-graph isomorphism which is \cdef{component sensitive}, that is 
\[
\textrm{for every }g\in\brt\tgla \textrm{ there is a } g'\in\vgt\tglb \textrm{ such that } i(g)=g'. \nonumber
\]
Then the \cdef{insertion} of $\tgla$ into $\tglb$ along $i$ is
\[
\tglb\circ_{i}\tgla := \left( \V_\tgla, \H_\tgla, \v_\tgla, \E_\tgla\cup i_\H^{-1}(\E_{\tglb}) ; \h_\tgla,\S_\tgla,\VS_\tgla, \ES_\tgla\cup i_\S^{-1}(\ES_{\tglb}) \right).
\] 
One denotes the set of possible insertions of $\tgla$ into $\tglb$, \ie the set of  component-sensitive 1-graph isomorphisms $i$, as $\ins \tgla{\tglb}$.
\end{defin}

\begin{remark}
The number of possible insertions is the number of component-sensitive \mbox{1-graph} automorphisms. 
As each such isomorphism is already fixed by the specific element in the orbit which it is mapping to, the only choice left is an additional automorphism. 
Thus one has
\[\label{eq:numberofinsertions}
|\ins \tgla{\tglb}| = |\aut\brt\tgla| = |\aut\vgt\tglb| 
=\prod_{\ogl\in\OG}V_{\tglb}^{\ogl}! \prod_{v\in\V_{\tglb}} |\aut\vg|
\]
where $V_{\tglb}^{\ogl}$ is the number of vertices in $\tglb$ with vertex graph $\ogl$.

Dropping the restriction to component sensitivity would allow to insert a 2-graph $\tglb$ into $\tgla$ with $\br\tgla\osim\vgs\tglb$ but $\brt\tgla\ncong\vgt\tglb$.
For example, a 2-graph with two boundary components such as $\cfish$ from \eqref{eq:r4fishgraph} with distinguished colours $c_1\ne c_2$ could then be inserted not only into a vertex of its residue type, the multi-trace vertex $\cvft$, but also in two copies of the vertex $\cvt$. 
While this is definitely an interesting operation from a topological point of view (basically the operation of adding a handle), it cannot be inverted in terms of a contraction. This is because it involves two vertices but contraction yields by definition only one vertex per connected component.
\end{remark}

With this crucial concept of component sensitivity the proof of the central identity is basically the same as for usual graphs (see e.g. \cite{Borinsky:2018}):
\begin{proof}
One can expand the coproduct of a 2-graph $\tg$ which is a sum over contractions of subgraphs equivalently as a sum over pairs of 2-graphs $(\sg,\ttg)$ whose insertion of one into the other yields $\tg$,
\[
\cop\tg \equiv \sum_{\sg'\sgr\tg}\sg'\otimes\tg/\sg' 
= \sum_{\sg,\ttg\in\TG} \left|\{\sg'\sgr\tg | \sg'\tsim\sg \textrm{ and } \tg/\sg'\tsim \ttg\} \right|\, \sg\otimes\ttg \, .
\]
How many such pairs $(\sg, \ttg)$ are there? 
On the level of their representatives, that is for labelled 2-graphs $\sgl,\tgl, \tilde{\tgl}$, the set of triples $(\sgl',j_1,j_2)$ of subgraphs $\sgl'\sgr\tgl$ with isomorphisms $j_1:\sgl\to\sgl'$ and $j_2:\tgl/\sgl'\to\tilde{\tgl}$ is isomorphic to the set of pairs $(i,j)$ of an insertion place $i\in\ins{\sgl}{\tilde{\tgl}}$ and an isomorphism $j:\tilde{\tgl}\circ_i \sgl\to\tgl$.%
\footnote{
The argument works along the lines of Lemma 5.3.1 of \cite{Borinsky:2018} with a straightforward generalization from 1-graphs to 2-graphs:
From $(\sgl',j_1,j_2)$ one can construct an insertion place $i'\in\ins{\sgl'}{\tgl/\sgl'}$ which is isomorphic to $i\in\ins{\sgl}{\tilde{\tgl}}$ via $j_1$ and $j_2$. 
Furthermore this gives an isomorphism $j:\tilde{\tgl}\circ_i \sgl\to\tgl$ since $\tgl/\sgl'\circ_{i'}\sgl'=\tgl$.
The other way, from $(i,j)$ one has 
$\sgl\sgr\tilde{\tgl}\circ_i \sgl$ 
such that there must be a unique $\sgl'\sgr j(\tilde{\tgl}\circ_i \sgl)=\tgl$ with an induced isomorphism $j_1$ such that $j_1(\sgl)=\sgl'$.
Contraction on both sides induces further the isomorphism $j_2$.
}
Therefore the cardinality of the two sets agrees,
\[
|\{\sgl'\sgr\tgl | \sgl'\tsim\sgl \textrm{ and } \tgl/\sgl'\tsim \tilde{\tgl}\}| \cdot|\aut\sgl| \cdot|\aut\tilde{\tgl}|
= |\{i\in\ins{\sgl}{\tilde{\tgl}}|\tilde{\tgl}\circ_i \sgl\tsim\tgl\}| \cdot|\aut\tgl| \, .
\]
Using this identity again on the level of unlabelled 2-graphs one finds
\begin{align}
   \cop X &= \sum_{\tg\in\TG}\frac1{|\aut\tg|} \sum_{\sg,\ttg\in\TG} \frac{|\{i\in\ins{\sg}{\ttg}|\ttg\circ_i \sg\tsim\tg\}|\,|\aut\tg|}{|\aut\sg|\,|\aut\ttg|} \sg\otimes\ttg 
\end{align}
One can change the sums and use for given $\sg, \ttg\in\TG$ the simple fact
\[
\sum_{\tg\in\TG} |\{i\in\ins{\sg}{\ttg}|\ttg\circ_i \sg\tsim\tg\}| = |\ins{\sg}{\ttg}|
\]
to find
\[
\cop X  = \sum_{\sg,\ttg\in\TG}|\ins{\sg}{\ttg}|\frac\sg{|\aut\sg|} \otimes \frac\ttg{|\aut\ttg|}
        = \sum_{\sg,\tg \in\TG}|\ins{\sg}{\tg} |\frac\sg{|\aut\sg|} \otimes \frac\tg {|\aut\tg|}\,
\]
(where the second step is simply a relabelling).

The symmetry-weighted sum over all 2-graphs $\sg$ with a suitable structure of connected components to be inserted into a 2-graph $\tg$ is 
$\prod_{v\in V_\tg}X^{\og_v}/\prod_{\og\in\OG}V_\tg^\og!$ 
in which the denominator takes care of implicit automorphisms between isomorphic 2-graphs with the same boundary. Inserting the number of insertions \eqref{eq:numberofinsertions} one has finally
\[
\cop X = \sum_{\tg\in\TG} \prod_{\og\in\OG}V_{\tg}^{\og}! \prod_{v\in\V_{\tg}} |\aut\og_v| \frac{\prod_{v\in V_\tg}X^{\og_v}}{\prod_{\og\in\OG}V_\tg^\og!} \otimes \frac\tg{|\aut\tg|}
\]
which concludes the proof.
\end{proof}

\section{The Hopf algebra of 2-graphs and renormalization in cNLFT}\label{sec:hopfalgebra}

In this section I show that for every renormalizable combinatorially non-local field theory there is a Connes-Kreimer type Hopf algebra of divergent Feynman diagrams \cite{Connes:2000km,Connes:2001hk,Kreimer:2002br}.
Besides the specific external structure of such diagrams being 2-graphs with a 1-graph boundary, one can use the same logic to construct these algebras as for local field theory \cite{Borinsky:2018}:
There is a general Hopf algebra of all 2-graphs which can be restricted to the specific set of diagrams of a given theory and further to the subset of divergent diagrams.
In the end I will illustrate how this works in the example of matrix field theory and for field theories with tensorial interactions.

\

By the same arguments as for 1-graphs, the bialgebra of 2-graphs $\btg$ has a unique coinverse element and is thus a Hopf algebra.
From the renormalization perspective, the important insight in this is that there is group structure on the set of algebra automorphisms on $\btg$ which extends also to algebra homomorphisms $\phi, \psi:\btg\to\uca$ mapping to any other unital commutative algebra~$\uca$.
The multiplication of such homomorphisms is given by the \emph{convolution product}
\[
\phi \conp \psi := m_{\uca} \circ (\phi\otimes\psi) \circ \cop
\]
which due to (co)associativity of $m_\uca$ and $\cop$ is also associative 
and has the neutral element $u_\uca\circ\cou_\btg$.
In particular, $\uca$ can be the algebra of amplitudes of a field theory and $\phi$ the evaluation of the amplitude labelled by a given Feynman diagram $\tg\in\btg$.
For renormalization, the object of interest is then the group inverse of this evaluation map of diagrams.

\begin{theorem}[Hopf algebra of 2-graphs/group of algebra homomorphisms]\label{th:hopfalgebra}
The bialgebra of 2-graphs $\btg$ is a Hopf algebra, i.e.
there exists an \emph{antipode} $\anti$, that is a unique inverse to the identity $\id:\tg\mapsto\tg$  with respect to the convolution product,
\[
\anti \conp \id = \id \conp \anti = u \circ \cou \, .
\]
Furthermore, the set $\Phi_{\uca}^{\btg}$ of algebra homomorpisms from $\btg$ to a unital commutative algebra~$\uca$ is a group with inverse $\anti^{\phi} = \phi\circ S$ for every $\phi\in\Phi_{\uca}^{\btg}$,
\[\label{eq:groupinverse}
\anti^\phi \conp \phi = \phi \conp \anti^\phi = u_\uca \circ \cou_\btg .
\]
The subbialgebra $\langle\overline{\TG(\Vtypeg)}\rangle$ generated by 2-graphs with specific vertex graphs $\Vtypeg\In\OG$ (Prop. \ref{prop:subbialgebra}) is a Hopf subalgebra of $\btg$.
\end{theorem}

\begin{proof}
The 1-graph structure contained in the 2-graphs is already sufficient and there are two standard ways to prove the unique existence of the antipode \cite{Manchon:2004vi} which both use the fact that the bialgebra $\btg$ is a graded bialgebra $\btg=\bigoplus_{E\ge0} \btg_E$ with respect to the number of edges (which follows directly from Def. \ref{def:contraction}).
The common first option is to
{reduce to the augmentation ideal 
$\textrm{Ker}\cou = \{\one\}\oplus\bigoplus_{E\ge1} \btg_E$} 
which is a connected ($\btg_0$ is one-dimensional) 
graded bialgebra, and thus automatically a Hopf algebra (Sec.\,III.3 in \cite{Manchon:2004vi}).
Here we use the second construction following \cite{Borinsky:2018} where one keeps all the residues $r\in\Res^*$ which behave as \emph{group-like elements} w.r.t. the coproduct \cite{Manchon:2004vi} , $\cop r=r\otimes r$, and augment $\btg$ by formal inverses $r^{-1}$ for all $\one\ne r\in\Res^*$ defining $r^{-1} r = r\,r^{-1}=\one$.
Then the antipode on $r\in\Res^*$ is $\anti_{\Res^*}(r):=r^{-1}$ from which one can construct the antipode on $\btg$ as 
\[
\anti:=\sum_{k\ge 0} \left(u\circ\cou-\anti_{\Res^*} \right)^{\conp k}\conp \anti_{\Res^*} \, 
\]
as proven for example in \cite{Borinsky:2018}.

For the inverse of a homomorphism $\phi:\btg\to\uca$ one directly calculates
\[
\anti^\phi \conp \phi = m_\uca\circ((\phi\circ\anti)\otimes\phi)\circ\cop = \phi\circ m_\uca\circ(\anti\otimes\id)\circ\cop = \phi\circ u_\btg\circ\cou_\btg = u_\uca\circ\cou_\btg\,.
\]
Finally, by definition of the antipode $m\circ(\id\otimes\anti)\circ\cop(\langle\overline{\TG(\Vtypeg)}\rangle)=u\circ\cou(\langle\overline{\TG(\Vtypeg)}\rangle)$. 
It follows that 
$\langle\overline{\TG(\Vtypeg)}\rangle \anti(\btg) = \anti(\langle\overline{\TG(\Vtypeg)}\rangle)\btg$ 
and thus
$\anti(\langle\overline{\TG(\Vtypeg)}\rangle)\In \langle\overline{\TG(\Vtypeg)}\rangle$.
\end{proof}

Connes-Kreimer-type Hopf algebras for specific renormalizable field theories follow from the general Hopf algebra of 2-graphs by a further restriction.
With Thm.\,\ref{th:hopfalgebra} one has already a Hopf subalgebra for a specific class of 
Feynman diagrams as for example (Ex.\,\ref{ex:bialgebra}) combinatorial maps 
in matrix theories or $\rk$-coloured 2-graphs 
in tensor-invariant theories.
One can always further restrict to one-particle irreducible (bridgeless) 2-graphs as a quotient Hopf algebra dividing out the Hopf ideal generated by one-particle reducible 2-graphs (\cite{Borinsky:2018} Ex.\,5.5.1).
For renormalization one restricts then further to the subset of superficially divergent diagrams which will turn out to be contraction-closed by definition, and thus will also be a Hopf algebra.

Concerning the structure of renormalization one can consider quantum field theory from a purely combinatorial perspective \cite{Cvitanovic:1983ta, Yeats:2008, Yeats:2017}.
In this spirit, a local quantum field theory is specified by its field content, the interactions and their weights as well as the spacetime dimension~$\std$ \cite{Yeats:2008,Borinsky:2018}.
For combinatorially non-local field theory
let us restrict here to a single field type (of which there could still be various, for example with different number $\hd$ of arguments \cite{Oriti:2015kv}). 
Then a combinatorially non-local scalar field theory is simply given by its interactions, their weights and a dimension%
\footnote{
This dimension $\gd$ is not the dimension of spacetime but simply the dimension of a single argument of the field.
It can have various meaning depending on the specific theory.
For example, a matrix field theory derived from $\std$-dimensional non-commutative field theory has $\gd=\std/2$ \cite{Grosse:2004bm, Grosse:2005ec,Hock:2020tg}.
In group field theory, $\gd$ is the dimension of the group $G$ which the field arguments are valued in and $G$ is related to the Lorentz group if the theory has an interpretation as a quantum theory of gravity \cite{DePietri:2000dt,Oriti:2003uw,Freidel:2005jy,Carrozza:2013uq}.
}:

\begin{defin}[comb. field theory]
A {combinatorial cNLFT} $T=(\OG^e,\OG^v,\w,\gd)$ consists of a dimension $\gd\in\N$, boundary graph sets $\OG^e\In\OG^v\In\OG$ and a weight map
\[
\w : \OG^e\cup\OG^v \to \Z
\]
where $\OG^v$, or equivalently $\Res^v\In\Res\In\TG$ under the bijection $\vgb$, is the set of interactions and $\OG^e$ (or $\Res^e$) is the set of propagators which are graphs with two vertices (or respectively 2-valent 2-graph vertices)%
\footnote{
In field theory it is usually not necessary to have the two-valent interactions included ($\Res^e\In\Res^v$);
Without restriction to the underlying theory, these are added here for the simple formula of contraction according to Def.\,\ref{def:contraction} which necessarily leads to such a vertex when contracting a propagator-type subgraph~\cite{Borinsky:2018}.
}.
Its Feynman diagrams are all 2-graphs with $\OG^v$-type vertices,  
\[
\TG^T :=\TG(\OG^v) = 
\left\{\tg\in\TG  \middle| \vgt\tg\in \pmset{\OG^v} \right\} \In \TG \,,
\]
and they generate a Hopf algebra $\btg_T:=\langle \overline{\TG^T} \rangle$ by closure with respect to contractions according to Thm.\,\ref{th:hopfalgebra}.
\end{defin}

\begin{defin}[renormalizability and the Hopf algebra of Feynman 2-graphs]
Let $T=(\OG^e,\OG^v,\w,\gd)$ be a combinatorial cNLFT. 
For a $\tg\in\TG^T$, the \cdef{superficial degree of divergence} is
\[\label{eq:superficialdegree}
\sdd(\tg) = \sum_{v\in\V_\tg} \w(\og_v) - \sum_{e\in\E_\tg} \w(\og_e) + \gd\cdot\nf
\]
where $\nf=|\Fi|$ is the number of internal faces (Def.\,\ref{def:faces}), 
$\og_v$ is the vertex graph of the vertex $v$, \eqref{eq:vertexgraph}, and $\og_e$ the ``vertex graph'' of the edge $e$, defined as the 1-graph with two vertices and as many edges between them as there are strands at the edge $e$ in the 2-graph~$\tg$.
The set of superficially divergent diagrams in $T$ are bridgeless 2-graphs whose nontrivial components are superficially divergent,
\[
\psd_T := \left\{\tg=\bigsqcup{}_{i\in I} \tg_i \in\TG^T  \textrm{ bridgeless } \middle| \forall i\in I: \textrm{if }\tg_i\not\in\Res \textrm{ then } \sdd(\tg_i)\ge 0 \right\} \, .
\]
The cNLFT $T$ is \emph{renormalizable} iff $\sdd$ depends only on the boundary of 2-graphs and for every connected $\tg\in\TG^T$: 
\[
\sdd(\tg) = \w(\br\tg) \, .
\]
Then $\hfd_T = \langle\psd_T \rangle$ is the Connes-Kreimer Hopf algebra of divergent Feynman 2-graphs of the theory $T$.
It is a Hopf subalgebra of $\btg_T$ and $\btg$ according to Prop.\,\ref{prop:subbialgebra} and Thm.\,\ref{th:hopfalgebra}
as it is contraction closed due to renormalizablity.
\end{defin}

This is the combinatorial structure of renormalization.
To perform actual renormalization of a cNLFT, a renormalization scheme has to be specified by choosing an appropriate linear operator $\rota$ on the algebra of amplitudes $\uca$ on which diagrams are mapped by the homomorphism $\phi:\hfd_T\to\uca$.
Based on the existence of the antipode $\anti$ this allows to recursively define the counterterm map 
\begin{align}
\anti^\phi_\rota: \hfd\to\uca ,\quad 
\anti^\phi_\rota(\one) &= \one ,\quad
\anti^\phi_\rota(\Gamma) = -\rota\phi(\tg) - \rota\bigg[ \sum_{\sg\sgr\tg}\anti^\phi_\rota(\sg)\phi(\tg/\sg)\bigg]
\end{align}
which is in the group $\Phi^{\hfd}_\uca$ if $\rota$ is a Rota-Baxter operator \cite{Manchon:2004vi}.
Then the renormalized Feynman amplitudes are evaluated by $\anti^\phi_\rota\conp\phi$.
I exemplify how this works for a cNLFT in practice in \cite{Thurigen:2021wn} and close this work with two examples of renormalizable cNLFT from the purely combinatorial perspective.


\begin{example}[Hopf algebra of matrix field theory/Grosse-Wulkenhaar model]
\label{ex:matrixfieldtheory}
The Feynman diagrams of matrix field theories \cite{Grosse:2005ec,Wulkenhaar:2019,Hock:2020tg} are combinatorial maps which are 2-graphs with bivalent edges, $\hd=2$ (Ex.\,\ref{ex:matrixdiagrams}). 
Thus, the interaction vertex graphs can be polygons with $n$ vertices and disjoint unions thereof (multi-trace vertices). 

So called \emph{$\phi_\std^n$ matrix field theory} is specified by $\OG^e=\{\edgevg\}$, a set of interactions $\OG^v$ 
of maximal order $n$ and a dimension $d=\std/2$ which can be derived as the spectral dimension of an external matrix (the kinetic operator in the action) \cite{Wulkenhaar:2019, Hock:2020tg} and relates to the spacetime dimension $\std$ of a corresponding non-commutative field theory for even~$\std$~\cite{Grosse:2005ec}.
Let the weights be $\w|_{\OG^e} = 1$ and $\w|_{\OG^v\setminus\OG^e} = 0$.

According to \eqref{eq:superficialdegree} the superficial degree of divergence is  $\sdd(\tg) = \gd\cdot\nf_\tg - \nei_\tg$ where
$\nei_\tg=|\E_\tg|$ is the number of (internal) edges and 
$\nf_\tg=|\Fi_\tg|$ is the number of internal faces.
The Euler formula for such a 2-graph $\tg$ (dual to a triangulated surface) is
\[\label{eq:Eulerformula}
2-2g_\tg-\nc_{\br\tg}  = \nv_\tg - (\nei_\tg+\nv_{\br\tg}) +\nf_\tg +\nei_{\br\tg} 
= \nv_\tg - \nei_\tg+ \nf_\tg 
\]
with genus~$g_\tg$ and $\nc_{\br\tg}$ the number of connected components of the boundary.
The number of boundary vertices/external edges $\nv_{\br\tg}=|\V_{\br\tg}|=|\Ex_\tg|$ and boundary edges/external faces $\nei_{\br\tg}=|\E_{\br\tg}|=|\Fx_\tg|$ agree 
since the boundary graphs are polygons. 
Together with the relation 
\[\label{eq:edgevertexrelation}
2\nei_\tg+\nv_{\br\tg} = \sum_{\og\in\OG^v} \nv_\og \cdot \nv_\tg^\og
= \sum_{k=1}^n k \nv^{(k)}_\tg
\]
between edges and the number $\nv_\tg^\og$ of interactions in $\tg$ with $\nv_\og$ boundary vertices, or $\nv^{(k)}_\tg$ interactions of order $\nv_\og=k$, one obtains
\[\label{eq:divergencedegreematrix}
\sdd(\tg) \equiv \gd \nf_\tg - \nei_\tg
= -\gd(\nv_\tg-1) + \frac{\gd-1}{2}\left(\sum_{k=1}^n k \nv^{(k)}_\tg - \nv_{\br\tg} \right)
-\gd\left(2g_\tg + \nc_{\br\tg}-1\right) \, .
\]
For example, the Grosse-Wulkenhaar model, a $\phi_\std^4$ matrix field theory (only $V^{(4)}_\tg\ne0$) related to $\std=2\gd$ dimensional non-commutative field theory, has \cite{Grosse:2004ey}
\[
{2}\sdd(\tg) = 
\std - \frac{\std-2}{2}\nv_{\br\tg} + (\std-4)\nv_\tg - \std(2g_\tg + \nc_{\br\tg} -1 ) \, .  
\]
For $\std=4$ the divergence degree is independent of $\nv_\tg$ and only planar maps $\tg$, that is genus $g_\tg=0$, with single boundary component $\nc_{\br\tg}=1$ and a maximal number of $\nv_{\br\tg}\le4$ boundary vertices/external edges can be divergent.
Then, divergence depends only on the external structure of diagrams and $\phi_4^4$ matrix theory is renormalizable. 
In particular, the only Green's functions which need renormalization are the effective propagator (2-point function) and the regular ($\nc_{\br\tg}=1$) 4-point function.
This defines the set of superficially divergent diagrams $\psd_\textsc{gw}$ 
which generates the Connes Kreimer Hopf algebra of divergent 2-graphs for the $\std=4$ Grosse-Wulkenhaar model. 

A first attempt to construct the Connes-Kreimer Hopf algebra of quartic non-com\-mutative field theory in $\std=4$ dimensions has been obtained already in \cite{Tanasa:2007xa} and further detailed in~\cite{Tanasa:2013kh}.
However, this earlier work is lacking the understanding that the external structure of 2-graphs are 1-graphs (here $\hd=2$ regular graphs, \ie polygons) and that it is in particular crucial to respect their number of connected components $\nc_{\br\tg}$.%
\footnote{In earlier work on non-commutative field theory the focus was on the concept of the number of ``broken faces'' $B$ lacking the understanding that this is actually the number of boundary components $\nc_{\br\tg}$. 
The issue becomes apparent with the definition of insertions in \cite{Tanasa:2013kh} which does not respect the boundary structure (see for example Fig.\,9 therein where a diagram with $\nc_{\br\tg}=2$ is inserted in a vertex $v$ with $\nc_{\gamma_v}=1$).}
The systematic approach here clarifies this issue and applies to renormalizable non-commutative field theory and matrix field theory of arbitrary interactions and corresponding to any dimension~$\std$.
\end{example}

\begin{example}[Hopf algebra of tensorial field theories]
\label{ex:tensorialfieldtheory}
The Feynman diagrams of complex tensorial field theories of rank $\rk\ge2$  \cite{BenGeloun:2014gp} are 2-graphs with $\rk$-valent edges, $\hd=\rk$, which can be represented as bipartite $(\rk+1)$-coloured graphs (up to multi-trace vertices, Ex.\,\ref{ex:colouredgraphs}). 
The interaction vertex graphs can be any $\rk$-coloured graphs.

A \emph{$\phi_{\gd,\rk}^n$ tensorial field theory} is specified by $\OG^e=\{\edgevgr\}$, \ie the propagator vertex graph with $\rk$ edges,
a set of bipartite $\rk$-coloured graphs $\OG^v$ with maximal number of vertices~$n$ and a dimension $d$.
The weights 
for the propagator are $\w|_{\OG^e} = \ks$, $\zeta>0$
\cite{BenGeloun:2014gp}.

Transforming the superficial degree of divergence $\sdd(\tg)=\gd\nf_\tg-\ks\nei_\tg+\sum_v\w(\og_v)$ into a form which is meaningful for the question of renormalizability is more involved than for matrices since the 2-graphs are now dual to simplicial pseudo manifolds of dimension $\rk$.
Generalizing the genus $g$, the relevant quantity for power counting turns out to be Gurau's degree $\gdeg$ of a coloured graph \cite{Gurau:2011he,Gurau:2012ek}.
The Gurau degree $\gdeg\in\N$ is not a topological invariant (except for $\rk=2$ where $\gdeg=g$) but induced from the genus of Heegaard splitting surfaces of the pseudo manifold \cite{Ryan:2012bg} (for $\rk=3$, and a generalized notion of splitting surfaces in $\rk\ge4$).
Since the boundary of an $(\rk+1)$-coloured graph is an $\rk$-coloured graph, and thus a pseudo manifold of dimension $\rk-1$, also the Gurau degree of the boundary plays a role. 
Applying the Euler formula \eqref{eq:Eulerformula} repeatedly for the various splitting surfaces and using the edge-vertex relation \eqref{eq:edgevertexrelation} one finds \cite{OusmaneSamary:2014bs,BenGeloun:2014gp}
\[
\sdd(\tg) 
= \frac{\uvd-\ks}{2}\bigg(\sum_{k} k \nv^{(k)}_\tg - \nv_{\br\tg} \bigg)
-\uvd(\nv_\tg-1)
-\gd\left(\frac{2 \gdeg_\tg-2\gdeg_{\br\tg}}{(\rk-1)!} + \nc_{\br\tg}-1\right)
+\sum_{v\in\OG^v}\w(\og_v)
\]
where $\uvd:=\gd(\rk-1)$ is thus the counterpart to the dimension of local field theory from the renormalization perspective, \ie $\phi^n_{\gd,\rk}$ tensorial field theory behaves similar to $\phi^n_{\std=\uvd}$ local field theory.
Note that the above example of matrices, $\rk=2$, is covered by this formula%
\footnote{
More precisely, the $\rk=2$ case is a complex matrix field theory, not Hermitian as in Ex.\,\ref{ex:matrixfieldtheory}. Still, it is known that the large-$N$ properties of complex and Hermitian matrix models agree \cite{Ambjorn:1990ib} and it is therefore no surprise that the same is true for renormalizability of their field theory versions.
}.

A vanishing Gurau degree gives the maximal contribution to the divergence degree.
One can prove that always $\gdeg_\tg\ge\gdeg_{\br\tg}$ \cite{BenGeloun:2013fw}. 
Thus, $-(\gdeg_\tg-\gdeg_{\br\tg})$ is maximal for $\gdeg_\tg=\gdeg_{\br\tg}=0$.
The question which $\phi^n_{\gd,\rk}$ tensorial theories are renormalizable is then as usual a question of how the maximal order of interactions $n$ is balanced by the dimension $\uvd$, with some more freedom given by the propagator weight $\ks$. 
The dependence on the number of 2-graph vertices vanishes,
\[
\sdd(\tg) 
= \uvd - \frac{\uvd-\ks}{2}\nv_{\br\tg}
-\gd\left(\frac{2 \gdeg_\tg-2\gdeg_{\br\tg}}{(\rk-1)!} + \nc_{\br\tg}-1\right)
\, ,\]
when setting the vertex weights for each vertex $v$ of degree $d^v$ to
\[
\w(\og_v) = \uvd - \frac{d^v}{2}(\uvd-\ks) \, .
\]
Thus, there can be super-renormalizable theories for $\uvd=\ks$ and otherwise the maximal degree for renormalizable interaction vertices is  $n=\lfloor\frac{2\uvd}{\uvd-\ks}\rfloor$ \cite{BenGeloun:2014gp}.
The theory is just renormalizable if the number in the floor bracket is integer.

The first example, 
and a particularly interesting one, is the BenGeloun-Rivasseau model \cite{BenGeloun:2013fw}, \ie $\phi^6_{1,4}$ tensorial field theory with $\zeta=1$.
Thus, it has divergence degree
\[
\sdd(\tg) 
= 6-\nv_{\br\tg} 
-\frac{1}{3}\left(\gdeg_\tg-\gdeg_{\br\tg}\right) - (\nc_{\br\tg}-1) \, .
\]
This means that among others also the 4-point Green's function with two boundary components needs renormalization. 
It is thus necessary to include the multi-trace vertex $\cvft$ in the set $\OG^v$ of interactions and all the subtleties discussed in Sec.\,\ref{sec:contraction} apply.
Furthermore, 
there are divergent diagrams with $\gdeg_\tg>0$.
Thus, there is some dependence on the bulk topology in this case.
Still, the set $\psd_\textsc{bgr}$ of divergent 2-graphs is contraction closed since one can prove that the Gurau degree is invariant under contractions (\cite{Raasakka:2013wa}, Lemma 4.1). 
Then one obtains the Connes-Kreimer Hopf algebra of divergent 2-graphs of the BenGeloun-Rivasseau model as
$\hfd_\textsc{bgr}
=\langle \psd_\textsc{bgr}\rangle$. 
In the same way this works also for a number of renormalizable tensorial theories with other rank $\rk$ and dimension $\gd$ \cite{BenGeloun:2014gp}.
\end{example}

\subsection*{Acknowledgments}
I thank M. Borinsky, J. Ben Geloun, A. Hock, D. Kreimer, A. Pithis and R. Wulkenhaar for discussions and comments on the paper, in particular M. Borinsky for very helpful discussions on the various subalgebra structures occuring in Hopf algebras of diagrams.
This work was funded by the Deutsche Forschungsgemeinschaft (DFG, German Research Foundation) in two ways,
primarily under the project number 418838388 and
furthermore under Germany's Excellence Strategy EXC 2044–390685587, Mathematics M\"unster: Dynamics–Geometry–Structure.

\bibliographystyle{JHEP}
\bibliography{algebra}

\end{document}